\documentclass[12pt,final]{amsart}
\usepackage{amssymb}
\usepackage{amsmath}
\usepackage{amsthm}
\usepackage{graphicx}
\usepackage{color}
\usepackage{url}
\usepackage{paralist}

\usepackage{cite}
\usepackage{subfigure}
\usepackage{showkeys}

\numberwithin{equation}{section}

\newcommand{\R}{\mathbb{R}}
\newcommand{\T}{\mathbb{T}}
\newcommand{\N}{\mathbb{N}}

\newtheorem{theorem}{Theorem}[section]
\newtheorem{lemma}[theorem]{Lemma}

\newtheorem{proposition}[theorem]{Proposition}
\newtheorem{corollary}[theorem]{Corollary}
\theoremstyle{definition}
\newtheorem{definition}[theorem]{Definition}

\newcommand{\abs}[1]{\left| #1\right|}
\newcommand{\norm}[1]{\left\|#1\right\|}
\newcommand{\paren}[1]{\left(#1\right)}
\newcommand{\bracket}[1]{\left[#1\right]}
\newcommand{\set}[1]{\left\{#1\right\}}
\newcommand{\inner}[2]{\left\langle #1,#2\right\rangle}
\newcommand{\tr}{\mathrm{tr}}

\title[Generalized DNLS Soliton Stability]{Stability of Solitary Waves for a
  Generalized Derivative Nonlinear Schr\"odinger Equation}

\author[Liu]{Xiao Liu}
\address{Department of Mathematics, University of Toronto}

\author[Simpson]{Gideon Simpson}
\address{School of Mathematics, University of Minnesota}
 \thanks{G.S. was supported by NSERC.  His contribution to this work
   was completed under the NSF PIRE grant OISE-0967140 and the DOE grant DE-SC0002085.}

\author[Sulem]{Catherine Sulem}
\address{Department of Mathematics, University of Toronto}
 \thanks{C.S. is  partially  supported by NSERC through grant number 46179-11.}


\date{\today}

\begin{document}

\maketitle

\begin{abstract}
  We consider a derivative nonlinear Schr\"odinger equation with a
  general nonlinearity. This equation has a two parameter family of
  solitary wave solutions. We prove orbital stability/instability
  results that depend on the strength of the nonlinearity and, in some
  instances, their velocity. We illustrate these results with
  numerical simulations.
\end{abstract}

\section{Introduction}
\label{s:intro}

The derivative nonlinear Schr\"{o}dinger (DNLS) equation
\begin{equation}
  \label{eqn:DNLS}
  i \partial_t u+\partial_x^2 u+i(|u|^{2} u)_x=0,
\end{equation}
is a nonlinear dispersive wave equation that appears in the long
wavelength approximation of Alfv{\'e}n waves propagating in a plasma
\cite{MJOLHUS1976,MIO1976,Passot1993}. Applying the gauge
transformation
\begin{equation}
  \label{e:gauge}
  \psi=u(x)\exp i \left\{\frac{1}{2}\int_{-\infty}^x|u(\eta)|^2d\eta\right\},
\end{equation}
equation \eqref{eqn:DNLS} has the form
\begin{equation}
  \label{eqn:DNLS2}
  i \partial_t \psi+\partial_x^2 \psi+i|\psi|^{2}\psi_x=0,\quad x\in\mathbb{R}.
\end{equation}
This equation has a Hamiltonian structure and can be written as
\[
\frac{\partial \psi}{\partial t} = -i E'(\psi)
\]
where the Hamiltonian is
\[
E\equiv\frac{1}{2}\int_{-\infty}^\infty|\psi_x|^2
dx+\frac{1}{4}\Im\int_{-\infty}^\infty |\psi|^{2}\bar{\psi}\psi_x dx.
\]
Equation \eqref{eqn:DNLS2} and some of its generalizations, also
appear in the modeling of ultrashort optical pulses \cite{Agrawal2006,
  Moses2007}.  Furthermore, the DNLS equation has the remarkable
property of being integrable by inverse scattering \cite{Kaup1978}.
It admits a two-parameter family of solitary wave solutions of the
form:
\begin{equation}
  \label{eqn:soln}
  u_{\omega, c}(x,t) = \varphi_{\omega,c}(x-ct)\exp i \left\{ \omega t
    +\frac{c}{2}(x-ct)- \frac{3}{4}\int^{x-ct}_{-\infty}\varphi_{\omega,c}^{2}(\eta)d\eta\right\},
\end{equation}
where $\omega > c^2/4$ and
\begin{equation}
  \label{eqn:phi}
  \varphi_{\omega,c}(y) = \sqrt{ \frac{(4\omega-c^2)}{\sqrt{\omega}(\cosh(\sigma\sqrt{4\omega-c^2}y)-\frac{c}{2\sqrt{\omega}})}}
\end{equation}
is the positive solution to
\begin{equation}
  \label{eqn:stationary_soln}
  - \partial_{y}^2\varphi_{\omega,c}  + (\omega -
  \frac{c^2}{4})\varphi_{\omega,c} +\frac{c}{2}|\varphi_{\omega,c}|^{2}\varphi_{\omega,c} - \frac{3}{16}|\varphi_{\omega,c}|^{4}\varphi_{\omega,c}  = 0.
\end{equation}
Guo and Wu \cite{Guo1995} showed that these solitary waves are {\it
  orbitally stable} if $c<0$ and $c^2<4\omega$. Colin and Ohta
\cite{Colin2006} subsequently extended the result, proving orbital
stability for all $c, c^2<4\omega$.

\begin{definition} Let $u_{\omega,c}$ be the solitary wave solution of
  \eqref{eqn:DNLS}. The solitary wave $u_{\omega,c}$ is
  {\emph{orbitally stable}} if, for all $\epsilon>0$, there exists
  $\delta>0$ such that if $\|u_0-u_{\omega,c}\|_{H^1}<\delta$, then
  the solution $u(t)$ of \eqref{eqn:DNLS} with initial data
  $u(0)=u_0$, exists globally in time and satisfies
  \[
  \sup_{t\geq0}\inf_{(\theta,y)\in\T \times
    \R}\|u(t)-e^{i\theta}u_{\omega,c}(t,\cdot-y)\|_{H^1}<\epsilon.
  \]
  Otherwise, $u_{\omega,c}$ is said to be \emph{orbitally unstable}.
\end{definition}

In an effort to understand the structural properties of DNLS, we study
an extension of \eqref{eqn:DNLS2} with general power nonlinearity
($\sigma>0$).
\begin{equation}
  \label{eqn:gDNLS}
  i \partial_t\psi+\partial_x^2\psi+i|\psi|^{2\sigma}\psi_x=0.
\end{equation}
Equation \eqref{eqn:gDNLS} also admits a two-parameter family of
solitary wave solutions,
\begin{equation}
  \label{eqn:soln}
  \psi_{\omega,c}(x,t) = \varphi_{\omega,c}(x-ct)\exp i\left\{ \omega t
    +\frac{c}{2}(x-ct)-\frac{1}{2\sigma +
      2}\int^{x-ct}_{-\infty}\varphi_{\omega,c}^{2\sigma}(\eta)d\eta\right\},
\end{equation}
where $\omega> c^2/4$ and
\begin{equation}
  \label{eqn:def_phi}
  \varphi_{\omega,c}(y)^{2\sigma} = \frac{(\sigma+1)(4\omega-c^2)}{2\sqrt{\omega}(\cosh(\sigma\sqrt{4\omega-c^2}y)-\frac{c}{2\sqrt{\omega}})}
\end{equation}
is the positive solution of
\begin{equation}
  \label{eqn:stationary_soln}
  - \partial_{y}^2\varphi_{\omega,c}  + (\omega -
  \frac{c^2}{4})\varphi_{\omega,c} +\frac{c}{2}|\varphi_{\omega,c}|^{2\sigma}\varphi_{\omega,c} - \frac{2\sigma + 1}{(2\sigma +
    2)^2}|\varphi_{\omega,c}|^{4\sigma}\varphi_{\omega,c}  = 0.
\end{equation}
It is convenient to define
\begin{equation}
  \label{eqn:phi&varphi}
  \phi_{\omega,c}(y)=\varphi_{\omega,c}(y) e^{i\theta_{\omega,c}(y)},
\end{equation}
with the traveling phase
\begin{equation}
  \label{e:traveling_phase}
  \theta_{\omega,c} (y)\equiv \frac{c}{2}y-\frac{1}{2\sigma+2}\int^{y}_{-\infty}\varphi_{\omega,c}^{2\sigma}(\eta)d\eta.
\end{equation}
Clearly,
\begin{equation}
  \label{eqn:psi&phi}
  \psi_{\omega, c}(x,t)=e^{i\omega t}\phi_{\omega,c}(x-ct),
\end{equation}
and the complex function $\phi_{\omega,c}(y)$ satisfies
\begin{equation}
  \label{eqn:phi}
  -\partial_y^2\phi_{\omega,c} +\omega\phi_{\omega,c} +ic\partial_x\phi_{\omega,c} - i|\phi_{\omega,c}|^{2\sigma}\partial_y\phi_{\omega,c}  =0, y\in \mathbb{R}.
\end{equation}

Provided there is no ambiguity, we write $\phi,\varphi$ for
$\phi_{\omega,c}, \varphi_{\omega,c}$ respectively.  Furthermore, we
only consider {\it admissible} values of $(\omega, c)$ satisfying the
conditions $\omega > \frac{c^2}{4}, \quad c \in \R.$
\subsection{Main Results}

We investigate the stability of solitary wave solutions
$\psi_{\omega,c}$ to the gDNLS equation.  This is determined by both
the value of $\sigma$ and the choice of the soliton parameters, $c$
and $\omega$. These results are {\it conditional} in the sense that
for $\sigma\neq 1$, we lack a suitable local well-posedness theory.
Throughout our study, we assume that given $\sigma >0$, and $\psi_0
\in H^1(\R)$, there exists a weak solution $ \psi \in C\paren{[0, T);
  H^1(\R)}$ of \eqref{eqn:gDNLS}, for $ T >0$, which satisfies
\begin{equation}
  \label{e:weak}
  \frac{d}{dt}\inner{\psi(\cdot,t)}{f} = \inner{E'(\psi(\cdot,t))}{-Jf}
\end{equation}
for appropriate test functions $f$. $E$ is the energy functional, and
$J$ is the symplectic operator. These are defined in Section
\ref{s:setup}.

Subject to this assumption, we have the following results:
\begin{theorem}
  \label{thm:main}
  For any admissible $(\omega,c)$ and $\sigma \geq 2$, the solitary
  wave solution $\psi_{\omega,c}(x,t)$ of \eqref{eqn:gDNLS} is
  orbitally unstable.
\end{theorem}

For $\sigma$ between 1 and 2, slow solitons, those with sufficiently
low $c$, will be stable while fast right-moving solitons will be unstable:
\begin{theorem}[Numerical]
  \label{thm:main2}
  For $\sigma \in (1,2)$, there exists $z_0 = z_0(\sigma)\in (-1,1)$
  such that:
  \begin{enumerate}[(i)]
  \item the solitary wave solution $\psi_{\omega,c}(x,t)$ of
    \eqref{eqn:gDNLS} is orbitally stable for admissible $(\omega,c)$
    satisfying $ c < 2 z_0 \sqrt{\omega}$.
  \item the solitary wave solution $\psi_{\omega,c}(x,t)$ of
    \eqref{eqn:gDNLS} is orbitally unstable for admissible
    $(\omega,c)$ satisfying $ c > 2 z_0 \sqrt{\omega}$.
  \end{enumerate}
\end{theorem}

Our last result concerns $\sigma < 1$, where all solitons are stable:
\begin{theorem}[Numerical]
  \label{thm:main3}
  For admissible $(\omega,c)$ and $0<\sigma< 1$, the solitary wave
  solution $\psi_{\omega,c}(x,t)$ of \eqref{eqn:gDNLS} is orbitally
  stable.
\end{theorem}
The endpoint, $\sigma =1$, corresponds to the cubic case, which has
already been studied in \cite{Colin2006,Guo1995}.

Theorems \ref{thm:main2} and \ref{thm:main3} are fully rigorous up to
the determination of the sign of a function of one variable, which is
parametrized by $\sigma$.  The number $z_0$ in Theorem
\ref{thm:main2} corresponds to a zero crossing.  This function,
defined below by \eqref{e:Fsigz}, includes improper integrals of
transcendental functions.

Theorem \ref{thm:main2} is particularly noteworthy for distinguishing
gNLS from the focusing nonlinear Schr\"odinger equation (NLS),
\begin{equation}
  \label{e:NLS}
  i \psi_t  + \Delta \psi + \abs{\psi}^{2\sigma} \psi = 0,\quad
  \psi:\R^{d+1} \to \mathbb{C}.
\end{equation}
Equation \eqref{eqn:gDNLS} is invariant under the scaling
transformation $ \psi_{\lambda}(x,t)=\lambda^{\frac{1}{2\sigma}}
\psi\left(\lambda x,\lambda^{2} t\right).$ This implies that its
critical Sobolev exponent is $ s_c = \frac{1}{2}- \frac{1}{2\sigma}.$
Hence, it is $L^2$-critical for $\sigma=1$, and it is
$L^2$-supercritical, energy subcritical for $\sigma >1$.  While NLS
only admits stable solitons in the $L^2$-subcritical regime, gDNLS
admits stable solitons not only in the critical regime, but also in
the supercritical one.

\subsection{Outline}

Our results are proven using the abstract functional analysis
framework of Grillakis, Shatah and Strauss,
\cite{Grillakis1987,Grillakis1990}; see
\cite{Weinstein1985,Weinstein1986a} for related results and
\cite{Sulem1999} for a survey.  The test for stability involves two
parts: (i) Counting the number of negative eigenvalues of the
linearized evolution operator $H_\phi$ near the solitary solution of
\eqref{eqn:gDNLS}, denoted $n(H_\phi)$; (ii) Counting the number of
positive eigenvalues of the Hessian of the scalar function
$d(\omega,c)$ built out of the action functional evaluated at the
soliton, denoted $p(d'')$.  We give explicit characterizations of
$H_\phi$ and $d(\omega,c)$ in Section \ref{s:setup}.  We then apply:
\begin{theorem}[Grillakis et al. \cite{Grillakis1987,Grillakis1990}]
  \label{thm:GSS}
  \begin{equation}
    \label{eqn:p(d)}
    p(d'')\leq n(H_\phi)
  \end{equation}
  Furthermore, under the condition that $d$ is non-degenerate at
  $(\omega,c)$:
  \begin{enumerate}[(i)]
  \item If $p(d'') = n(H_\phi)$, the solitary wave is orbitally
    stable;
  \item If $n(H_\phi)-p(d'')$ is odd, the solitary wave is orbitally
    unstable.
  \end{enumerate}
\end{theorem}
In Section \ref{sec:spec_H}, we show that $n(H_\phi)=1$, and in
Section \ref{sec:concavity_d}, we compute the number of positive
eigenvalues of the Hessian, $d''$.  We complete the proofs of Theorems
\ref{thm:main} to \ref{thm:main3} in Section \ref{sec:Orbital_stabi}.
In Section \ref{s:disc}, we discuss the results and illustrate them
with numerical simulations of both stable and unstable solitons.
Algebraic manipulations useful for the evaluation of
$\det[d''(\omega,c)]$ and $\tr[d''(\omega,c)]$ are presented in the
Appendix.

\subsection{Remarks on Well Posedness Assumption}

The DNLS equation (cubic nonlinearity) has been studied in $H^1$ and
higher regularity spaces, with results for both local and global well
posedness results,
\cite{Hayashi1992,Hayashi1993,Hayashi1994,Hayashi1994a, Ozawa:1996uj,
  Tsutsumi1980,Tan:1994ju}. Much of the analysis relies on a
transformation related to \eqref{e:gauge} that turns the equation into two coupled semilinear
Schr\"odinger equations with no derivative.  In addition, Hayashi
\cite{Hayashi1992,Hayashi1993,Hayashi1994} identified a smallness
condition on the data
\begin{equation}
  \label{e:smallness}
  \|u_0\|_{L^2}<\sqrt{2\pi},
\end{equation}
for which global solutions exist in $H^s$, $s\in \N$.  The constant
$\sqrt{2\pi}$ is the $L^2$-norm of the ground state of the quintic NLS
soliton.  More recently, the global in time result for data satisfying
\eqref{e:smallness} were extended to $H^s$ spaces with $s > 1/2$ in \cite{Colliander2002}.  DNLS with low regularity has also been
studied on the torus, \cite{Grunrock:2005tk,Grunrock:2008go}.

There has also been progress beyond the cubic equation in the
aforementioned results.  Some studies, such as
\cite{Ozawa:1996uj,Tan:1994ju,Difranco2008}, include additive terms to
the cubic nonlinearity with derivative. More generally, Kenig, Ponce and Vega \cite{Kenig:1993wr,Kenig:1998wr}
used viscosity methods to show that, for general quasilinear
Schr\"odinger with polynomial nonlinearities, local well-posedness
holds in Sobolev spaces of high enough index (See Linares- Ponce
\cite{Linares2009} for a review). In \cite{Hao2007}, Hao proved
that \eqref{eqn:DNLS2} is locally well-posed in $H^{1/2}$
intersected with an appropriate Strichartz space
for $\sigma \geq 5/2$.  Working in the Schwartz space, Lee
\cite{Lee1989} used the framework of inverse scattering to show that
DNLS is globally well-posed for a dense subset of initial conditions,
excluding certain non generic ones.

\section{Problem Setup}
\label{s:setup}
In this section, we define the linearized operator, the invariant
quantities and the action functional of the soliton.  Throughout, we
adopt the notation of \cite{Grillakis1990}.

We study the problem in the space ${X = H^1(\mathbb{R})}$, with real
inner product
\begin{equation}
  \label{e:Xip}
  (u,v) \equiv \Re \int_{\mathbb{R}} (u_x \bar{v}_x + u \bar{v}) dx.
\end{equation}
The dual of $X$ is $X^* = H^{-1}(\mathbb{R})$. Let $I:X\rightarrow
X^*$ be the natural isomorphism defined by
\begin{equation}
  \langle I u,v \rangle = (u,v),
\end{equation}
where $\langle \cdot,\cdot \rangle $ is the pairing between $X$ and
$X^*$,
\begin{equation}
  \label{e:Xpairing}
  \langle f, u\rangle = \Re \int_{\mathbb{R}} f \bar{u} dx.
\end{equation}
gDNLS can be formulated as the Hamiltonian system
\begin{equation}
  \label{eqn:Hamiltonion}
  \frac{d\psi}{d t}=J E'(\psi),
\end{equation}
where the map $J:X^*\rightarrow X$ is $J=-i$, and $E$ is the
Hamiltonian:
\begin{equation}
  \label{eqn:Energy}
  E\equiv\frac{1}{2}\int_{-\infty}^\infty|\psi_x|^2
  dx+\frac{1}{2(\sigma+1)}\Im\int_{-\infty}^\infty
  |\psi|^{2\sigma}\bar{\psi}\psi_x dx. \quad \text{(Energy)}
\end{equation}
Two other conserved quantities are:
\begin{align}
  \label{eqn:Mass_dnls}
  Q\equiv\frac{1}{2}\int_{-\infty}^\infty|\psi|^2 dx, &\quad \text{(Mass)}\\
  \label{eqn:Momentum_dnls}
  P\equiv-\frac{1}{2}\Im\int_{-\infty}^\infty\bar{\psi}\psi_x dx.
  &\quad \text{(Momentum)}
\end{align}
Let $T_1$ and $T_2$ be the one-parameter groups of unitary operator on
$X$ defined by
\begin{subequations}
  \begin{align}
    T_1(t)\Phi(x)&\equiv e^{-i t}\Phi(x),\quad \Phi(x)\in X,\quad  t \in \mathbb{R}\\
    T_2( t)\Phi(x)&\equiv\Phi(x+t).
  \end{align}
\end{subequations}
Then
$$\psi_{\omega,c}=T_1(- \omega t) T_2(-c t) \phi_{\omega,c}(x).$$
and
\begin{equation*}
  \partial_t T_1(\omega t)|_{t=0}=-i\omega,\; \;
  \partial_t T_2(c t)|_{t=0}=c\partial_x ~.
\end{equation*}
We define the linear operators
\begin{subequations}
  \begin{align}
    B_1&\equiv J^{-1}\partial_t T_1(\omega t)|_{t=0}=\omega,\\
    B_2&\equiv J^{-1}\partial_t T_2(c t)|_{t=0}=ic\partial_x.
  \end{align}
\end{subequations}
The mass and momentum invariants \eqref{eqn:Mass_dnls} and
\eqref{eqn:Momentum_dnls} are thus related to the symmetry groups via:
\begin{subequations}
  \begin{align}
    \label{eqn:Q1}
    Q_1&\equiv\frac{1}{2}\langle B_1\phi,\phi\rangle
    =\frac{1}{2}\Re\int \omega\phi\bar{\phi}
    dx=\frac{\omega}{2}\int|\phi|^2 dx=\omega
    Q,\\
    \label{eqn:Q2}
    Q_2&\equiv\frac{1}{2}\langle B_2\phi,\phi\rangle
    =\frac{1}{2}\Re\int ic\partial_x \phi\bar{\phi}
    dx=-\frac{1}{2}c\Im\int\partial_x\phi \bar{\phi} dx=c P.
  \end{align}
\end{subequations}
Computing the first variations of \eqref{eqn:Energy},
\eqref{eqn:Mass_dnls} and \eqref{eqn:Momentum_dnls}, we have
\begin{subequations}
  \begin{align}
    E'(\phi)&=-\partial^2_x\phi-i|\phi|^{2\sigma}\partial_x\phi,\\
    Q'(\phi)&=\phi,\; \; \; P'(\phi)=i\partial_x\phi.
  \end{align}
\end{subequations}
The second variations are:
\begin{subequations}
  \label{eqn:E''_Q''_P''}
  \begin{align}
    \label{eqn:E''}
    \begin{split}
      E''(\phi)v&=(-\partial^2_x-i\sigma|\phi|^{2\sigma-2}\bar{\phi}\partial_x\phi
      -i|\phi|^{2\sigma}\partial_x)v\\
      &\quad -i\sigma|\phi|^{2\sigma-2}\phi\partial_x\phi \bar{v},
    \end{split}\\
    \label{eqn:Q''}
    Q''(\phi)v&=v,\;\;\; P''(\phi)v=i\partial_x v.
  \end{align}
\end{subequations}

\subsection{Linearized Hamiltonian}

The linearized Hamiltonian about the soliton $\phi$ is
\begin{equation}
  \label{eqn:H}
  \begin{split}
    H_{\phi}u&\equiv\bracket{E''(\phi)+Q''_1(\phi)+Q''_2(\phi)}u\\
    &\quad = \bracket{E''(\phi) + \omega Q''(\phi) + c P''(\phi)}u
  \end{split}
\end{equation}
for $u\in H^2(\mathbb{R})$.  For later use, we give two equivalent
expressions of $H_\phi$.  First, we decompose it into complex
conjugates:
\begin{lemma}
  \label{l:Hcc}
  For any function $u$ in the domain of $H_\phi$,
  \begin{equation*}
    H_\phi u = L_1 u + L_2 \bar u,
  \end{equation*}
  where
  \begin{subequations}
    \begin{align}
      L_1&\equiv-\partial^2_x+\omega+ic\partial_x-i\sigma|\phi|^{2\sigma-2}\bar{\phi}\partial_x\phi
      -i|\phi|^{2\sigma}\partial_x,\\
      L_2&\equiv-i\sigma|\phi|^{2\sigma-2}\phi\partial_x\phi.
    \end{align}
  \end{subequations}

\end{lemma}
Second, we give the expression of $H_\phi$ and the quadratic form it
induces, after extraction of the soliton's phase:
\begin{lemma}
  \label{cor:H}
  Let $u\in H^1$ be decomposed as
  \begin{equation}
    \label{eqn:p}
    u = e^{i\theta}(u_1+iu_2),
  \end{equation}
  where $\theta$ is given by \eqref{e:traveling_phase}, and $u_1$ and
  $u_2$ are the real and imaginary parts of $ue^{-i\theta}$.  Then
  \begin{equation}
    \label{e:Hdephased}
    H_\phi u = e^{i\theta} \bracket{\paren{L_{11}u_1  +L_{21}u_2} + i  \paren{L_{12}u_1  +L_{22}u_2}}
  \end{equation}
  and
  \begin{equation}
    \label{e:Hdephased_quad}
    \begin{split}
      \langle H_\phi u, u \rangle& = \langle
      L_{11}u_1,u_1\rangle+\langle
      L_{21}u_2,u_1\rangle\\
      &\quad +\langle L_{12}u_1,u_2\rangle+\langle L_{22}
      u_2,u_2\rangle,
    \end{split}
  \end{equation}
  where
  \begin{subequations}
    \label{e:dephased_blockops}
    \begin{align}
      L_{11}&\equiv-\partial_{yy}+\omega-\frac{c^2}{4}+\frac{c(2\sigma+1)}{2}\varphi^{2\sigma}-\frac{4\sigma^2+6\sigma+1}{4(\sigma+1)^2}\varphi^{4\sigma},\\
      L_{21}&\equiv-\frac{\sigma}{\sigma+1}\varphi^{2\sigma-1}\varphi_y+\frac{\sigma}{\sigma+1}\varphi^{2\sigma}\partial_y\\
      L_{12}&\equiv-\frac{(2\sigma+1)\sigma}{\sigma+1}\varphi^{2\sigma-1}\varphi_y-\frac{\sigma}{\sigma+1}\varphi^{2\sigma}\partial_y,\\
      L_{22}&\equiv-\partial_{yy}+\omega-\frac{c^2}{4}+\frac{c}{2}\varphi^{2\sigma}-\frac{2\sigma+1}{4(\sigma+1)^2}\varphi^{4\sigma}.
    \end{align}
  \end{subequations}
\end{lemma}
\begin{proof}
  Using \eqref{e:traveling_phase},
  \begin{equation}
    \begin{split}
      \theta_y=\frac{c}{2}-\frac{1}{2\sigma+2}\varphi^{2\sigma}, \quad
      \theta_{yy}=-\frac{\sigma}{\sigma+1}\varphi^{2\sigma-1}\varphi_y.
    \end{split}
  \end{equation}
  We can then rewrite the operators $L_1$ and $L_2$ from Lemma
  \ref{l:Hcc} in terms of $\varphi$ as
  \begin{align*}
    \begin{split}
      L_1
      &=-\partial^2_y+\omega+ic\partial_y-i\varphi^{2\sigma}\partial_y\\
      &\quad +\frac{\sigma
        c}{2}\varphi^{2\sigma}-\frac{\sigma}{2\sigma+2}\varphi^{4\sigma}-i\sigma\varphi^{2\sigma-1}\varphi_y,
    \end{split}\\
    L_2&=\left[\frac{c\sigma}{2}\varphi^{2\sigma}
      -\frac{\sigma}{2\sigma+2}\varphi^{4\sigma}-i\sigma\varphi^{2\sigma-1}\varphi_y\right]e^{2i\theta}.
  \end{align*}

  Letting $\chi = ue^{-i\theta}$, we have
  \begin{equation}
    \begin{split}
      H_{\phi}u=&e^{i\theta}\left[-\partial_{yy}+\omega-\frac{c^2}{4}+\frac{c(\sigma+1)}{2}\varphi^{2\sigma}-\frac{i\sigma^2}{\sigma+1}\varphi^{2\sigma-1}\varphi_y\right.\\
      &\left.-\frac{i\sigma}{\sigma+1}\varphi^{2\sigma}\partial_y-\frac{2\sigma^2+4\sigma+1}{4(\sigma+1)^2}\varphi^{4\sigma}\right
      ]\chi\\
      &+e^{i\theta}\left[\frac{c\sigma}{2}\varphi^{2\sigma}-\frac{\sigma}{2\sigma+2}\varphi^{4\sigma}-i\sigma\varphi^{2\sigma-1}\varphi_y\right]\bar{\chi}.
    \end{split}
  \end{equation}
  Since $\chi = u_1 + i u_2$, we get \eqref{e:Hdephased} and
  \eqref{e:Hdephased_quad} by grouping terms appropriately.
\end{proof}

\subsection{Scalar Soliton Function}
\label{s:def_d}
Using \eqref{eqn:phi}, we observe that when evaluated at the soliton $\phi$,
\begin{equation}
  \label{eqn:E'+Q1'+Q2'}
  E'+\omega Q'+ c P'=0.
\end{equation}
For any $\omega>c^2/4$, we define the scalar function
\begin{equation}
  \label{eqn:d_w_c}
  d(\omega,c)\equiv E(\phi_{\omega,c})+Q_1(\phi_{\omega,c})+Q_2(\phi_{\omega,c}),
\end{equation}
which is the action functional evaluated at the soliton.  It has the
following properties:
\begin{lemma}
  \label{lemma:d}
  \begin{align}
    \label{eqn:d_w_c_1}
    d(\omega,c)&=E(\phi)+\omega Q(\phi)+c P(\phi),\\
    \label{eqn:d_w=Q}
    \partial_{\omega} d(\omega,c)&=Q(\phi)>0,\\
    \label{eqn:d_c=P}
    \partial_c d(\omega,c)&=P(\phi).
  \end{align}
  The Hessian is
  \begin{equation}
    \label{eqn:d''}
    d''(\omega,c) =
    \begin{pmatrix}
      \partial_\omega Q(\phi) & \partial_c Q(\phi) \\
      \partial_\omega P(\phi) & \partial_c P(\phi)
    \end{pmatrix}.
  \end{equation}

\end{lemma}

\begin{proof}

  Using \eqref{eqn:Q1}, \eqref{eqn:Q2} and \eqref{eqn:d_w_c}, we have
  \eqref{eqn:d_w_c_1}. Differentiating \eqref{eqn:d_w_c_1} with
  respect to $\omega$ and $c$ respectively and using
  \eqref{eqn:E'+Q1'+Q2'}, we obtain \eqref{eqn:d_w=Q} and
  \eqref{eqn:d_c=P}.  The expression for the Hessian follows.

\end{proof}

\section{Spectral Decomposition of the Linearized Operator}
\label{sec:spec_H}

This section provides a full description of the spectrum of the
linearized operator $H_\phi$.  In particular, we prove:

\begin{theorem}
  \label{thm:Spectral_H}
  For all values of $\sigma>0$ and admissible $(\omega,c)$, the space
  $X= H^1$ can be decomposed as the direct sum
  \begin{equation}
    \label{eqn:X}
    X=N+Z+P,
  \end{equation}
  where the three subspaces intersect trivially and:
  \begin{enumerate}[(i)]
  \item $N$ is a one dimensional subspace such that for $u \in N$,
    $u\neq 0$,
    \begin{equation}
      \langle H_\phi u, u \rangle <0.
    \end{equation}

  \item $Z$ is the two dimensional kernel of $H_\phi$.

  \item $P$ is a subspace such that for $p \in P$,
    \begin{equation}
      \label{eqn:H_p}
      \langle H_\phi p, p \rangle \geq \delta \| p\|_X^2
    \end{equation}
    where the constant $\delta>0$ is independent of $p$.

  \end{enumerate}

\end{theorem}
\begin{corollary}
  For all values of $\sigma>0$ and admissible $(\omega,c)$,
  \label{c:Hnegcount}
  $$n(H_\phi) = 1.$$
\end{corollary}
An important ingredient of the proof involves rewriting the quadratic
form \eqref{e:Hdephased_quad} induced by $H_\phi$ in a more favorable
form.  This rearrangement, inspired by \cite{Guo1995}, expresses it as
a sum of a quadratic form involving an operator with exactly one
negative eigenvalue and a nonnegative term.
\begin{lemma} Let
  \[
  u = e^{i\theta}(u_1 + iu_2),
  \]
  where $\theta, u_1, u_2$ are defined the same as Corollary
  \ref{cor:H}, then
  \begin{equation}
    \label{eqn:H_varphi}
    \begin{split}
      \langle H_\phi u, u \rangle = \langle \widetilde{L}_{11} u_1,
      u_1 \rangle + \int_{-\infty}^\infty \bracket{\varphi \paren{
          \varphi^{-1} u_2}_y + \frac{\sigma}{(\sigma+1)}
        \varphi^{2\sigma} u_1}^2 dy,
    \end{split}
  \end{equation}
  where
  \begin{equation}
    \begin{split}
      \widetilde{L}_{11} \equiv
      &-\partial_{yy}+\omega-\frac{c^2}{4}+\frac{c(2\sigma+1)}{2}\varphi^{2\sigma}
      -\frac{8\sigma^2+6\sigma+1}{4(\sigma+1)^2}\varphi^{4\sigma}.
    \end{split}
  \end{equation}
\end{lemma}
\begin{proof}
  Recall the terms in the quadratic form
  \eqref{e:dephased_blockops}. We first examine $L_{11}$.  The
  relationship between $L_{11}$ and $\widetilde{L}_{11}$ is
  \begin{equation}
    \label{e:L11block}
    L_{11} = \widetilde{L}_{11} +\frac{\sigma^2}{(\sigma+1)^2} \varphi^{4\sigma}
  \end{equation}
  Next, consider $L_{22}$. From \eqref{eqn:stationary_soln},
  \begin{equation}
    \label{eqn:L22=0}
    L_{22} \varphi =0.
  \end{equation}
  Letting $\tilde u_2 = \varphi^{-1} u_2$, we can then write
  \begin{equation}
    \label{eqn:B2bb}
    \begin{split}
      \langle L_{22} u_2, u_2\rangle& = \langle -\partial_{yy} u_2,
      u_2 \rangle + \left\langle
        (\omega-\frac{c^2}{4}+\frac{c}{2}\varphi^{2\sigma}
        -\frac{2\sigma+1}{4(\sigma+1)^2}\varphi^{4\sigma}) u_2, u_2
      \right \rangle\\
      & = \langle -\varphi_{yy} \tilde u_2 - 2 \varphi_y
      \tilde u_{2y} - \varphi \tilde u_{2yy}, \varphi \tilde u_2 \rangle \\
      &\quad \quad +
      \langle (\omega-\frac{c^2}{4}+\frac{c}{2}\varphi^{2\sigma} -\frac{2\sigma+1}{4(\sigma+1)^2}\varphi^{4\sigma}) \varphi \tilde u_2, \varphi \tilde u_2 \rangle\\
      &= \langle \tilde u_2 L_{22} \varphi, \varphi \tilde u_2\rangle
      + \langle  - 2 \varphi_y \tilde u_{2y} - \varphi \tilde u_{2yy}, \varphi \tilde u_2 \rangle\\
      & = \langle -(\varphi^2 \tilde u_{2y})_y, \tilde u_2 \rangle =
      \langle \varphi \tilde u_{2y}, \varphi \tilde u_{2y} \rangle,
    \end{split}
  \end{equation}
  where $u_{2y}$ and $u_{2yy}$ denote $\partial_y u_2$ and
  $\partial_{yy} u_2$, respectively.  Lastly, we simplify the off
  diagonal entries, $L_{21}$ and $L_{12}$.  Integrating by parts, we
  have
  \begin{equation*}
    \begin{split}
      \langle L_{12}u_1,u_2\rangle
      =&\left\langle\left (-\frac{(2\sigma+1)\sigma}{\sigma+1}\varphi^{2\sigma-1}\varphi_y-\frac{\sigma}{\sigma+1}\varphi^{2\sigma}\partial_y\right)u_1,u_2\right\rangle \\
      =&-\frac{2\sigma+1}{2(\sigma+1)} \left\langle (\varphi^{2\sigma})_y, u_1 u_2 \right\rangle -\frac{\sigma}{\sigma+1} \left\langle \varphi^{2\sigma} u_{1y},u_2\right\rangle \\
      =&\frac{2\sigma+1}{2(\sigma+1)} \left \langle
        \varphi^{2\sigma}u_{2y},u_1\right\rangle
      +\frac{2\sigma+1}{2(\sigma+1)} \left\langle
        \varphi^{2\sigma}u_{1y},u_2\right\rangle
      -\frac{\sigma}{\sigma+1} \left\langle \varphi^{2\sigma}u_{1y},u_2\right\rangle \\
      =&\frac{2\sigma+1}{2(\sigma+1)}\left\langle
        \varphi^{2\sigma}u_{2y},u_1\right\rangle
      +\frac{1}{2(\sigma+1)} \left\langle
        \varphi^{2\sigma}u_{1y},u_2\right\rangle.
    \end{split}
  \end{equation*}
  Similarly,
  \begin{equation*}
    \begin{split}
      \langle L_{21}u_2,u_1\rangle =&\left\langle \left
          (-\frac{\sigma}{\sigma+1}\varphi^{2\sigma-1}\varphi_y+\frac{\sigma}{\sigma+1}\varphi^{2\sigma}\partial_y\right
        )u_2,u_1\right\rangle \\
      =&-\frac{1}{2(\sigma+1)} \left\langle (\varphi^{2\sigma})_y, u_1
        u_2 \right\rangle
      +\frac{\sigma}{\sigma+1} \left\langle \varphi^{2\sigma}           u_{2y},u_1\right\rangle \\
      =&\frac{1}{2(\sigma+1)} \left\langle
        \varphi^{2\sigma}u_{2y},u_1\right\rangle
      +\frac{1}{2(\sigma+1)} \left\langle
        \varphi^{2\sigma}u_{1y},u_2\right\rangle
      +\frac{\sigma}{\sigma+1} \left\langle \varphi^{2\sigma}u_{2y},u_1\right\rangle \\
      =&\frac{2\sigma+1}{2(\sigma+1)} \left\langle
        \varphi^{2\sigma}u_{2y},u_1\right\rangle
      +\frac{1}{2(\sigma+1)} \left \langle
        \varphi^{2\sigma}u_{1y},u_2\right\rangle .
    \end{split}
  \end{equation*}
  The off diagonal terms then sum to
  \begin{equation*}
    \langle L_{12}u_1,u_2\rangle + \langle L_{21}u_2,u_1\rangle  = \frac{2\sigma+1}{\sigma+1} \left\langle
      \varphi^{2\sigma}u_{2y},u_1\right\rangle
    +\frac{1}{\sigma+1} \left \langle \varphi^{2\sigma}u_{1y},u_2\right\rangle.
  \end{equation*}
  Introducing $\tilde u_2 = \varphi^{-1} u_2$ into the above
  expression, and integrating by parts,
  \begin{equation}
    \label{e:offdiagblocks}
    \langle L_{12}u_1,u_2\rangle + \langle L_{21}u_2,u_1\rangle  =\frac{2\sigma}{\sigma+1}  \left \langle
      \varphi ^{2\sigma+1} \tilde u_{2y} , u_1\right\rangle
  \end{equation}
  Combining \eqref{e:L11block}, \eqref{eqn:B2bb} and
  \eqref{e:offdiagblocks},
  \begin{equation*}
    \begin{split}
      \langle H_\phi u, u \rangle =& \langle \widetilde{L}_{11} u_1,
      u_1 \rangle + \langle \varphi
      \tilde u_{2y}, \varphi \tilde u_{2y} \rangle \\
      &\quad + \left\langle \frac{\sigma}{\sigma+1} \varphi^{2\sigma}
        u_1, \frac{\sigma}{\sigma+1} \varphi^{2\sigma} u_1
      \right\rangle + \left\langle \frac{2\sigma}{\sigma+1} \varphi
        ^{2\sigma} u_1 ,
        \varphi \tilde u_{2y} \right\rangle\\
      =& \langle \widetilde{L}_{11} u_1, u_1 \rangle +
      \int_{-\infty}^\infty \bracket{\varphi \tilde u_{2y} +
        \frac{\sigma}{\sigma+1} \varphi^{2\sigma} u_1}^2 dy.
    \end{split}
  \end{equation*}
\end{proof}

      \subsection{The Negative Subspace}
      Next, we characterize the negative subspace, $N$.  For that, we
      need the following lemma on $\widetilde{L}_{11}$.
      \begin{lemma}
        \label{l:spec_L11}
        The spectrum of $\widetilde{L}_{11}$ can be characterized as
        follows:
        \begin{itemize}
        \item $\widetilde{L}_{11}$ has exactly one negative
          eigenvalue, denoted $-\lambda_{11}^2$, with multiplicity
          one, and eigenfunction $\chi_{11}$,
        \item $0 \in \sigma(\widetilde{L}_{11})$, and the kernel is
          spanned by $\varphi_y$,
        \item There exists $\mu_{11} >0$ such that
          \[
          \sigma(\widetilde{L}_{11}) \setminus\set{-\lambda_{11}^2, 0}
          \subset [\mu_{11}, \infty).
          \]
        \end{itemize}

      \end{lemma}
      \begin{proof}

        First, we observe that since $\varphi$ is exponentially
        localized, $\widetilde{L}_{11}$ is a relatively compact
        perturbation of $ - \partial_y^2 + \omega - \tfrac{c^2}{4}.  $
        By Weyl's theorem, the essential spectrum is then
        \[
        \sigma_{\rm ess} (\widetilde{L}_{11}) = \sigma_{\rm ess}
        (- \partial_y^2 + \omega - \tfrac{c^2}{4})= \left [ \omega -
          \tfrac{c^2}{4}, \infty\right).
        \]
        Consequently, all eigenvalues below the lower bound of the
        essential spectrum correspond to isolated eigenvalues of
        finite multiplicity.  By differentiating
        \eqref{eqn:stationary_soln} with respect to $y$, we see that
        \begin{equation}
          \label{eqn:L11}
          \widetilde{L}_{11} \varphi_y =0.
        \end{equation}
        Hence, $\widetilde{L}_{11}$ has a kernel.  Viewed as a linear
        second order ordinary differential equation,
        $\widetilde{L}_{11}f =0$ has two linearly independent
        solutions.  As $y\to -\infty$, one solution decays
        exponentially while the other grows exponentially.  Thus, up
        to a multiplicative constant, there can be at most one
        spatially localized solution to $\widetilde{L}_{11}f =0$.
        Therefore, the kernel is spanned by $\varphi_y$.

        From Sturm-Liouville theory, this implies that zero is the
        second eigenvalue of $\widetilde{L}_{11}$, and
        $\widetilde{L}_{11}$ has exactly one strictly negative
        eigenvalue, $-\lambda_{11}^2$, with a $L^2$ normalized
        eigenfunction $\chi_{11}$:
        \begin{equation}
          \label{eqn:L11_chi11}
          \widetilde{L}_{11}\chi_{11} = -\lambda_{11}^2 \chi_{11}.
        \end{equation}
        If we now let
        \begin{equation}
          \label{e:L11bnd}
          \mu_{11} \equiv \inf_{f \neq 0, f\perp \varphi_y, f\perp \chi_{11}} \frac{\inner{\widetilde{L}_{11} f}{f}}{\inner{f}{f}}
        \end{equation}
        we see that $\mu_{11} >0$, since if it were not, it would
        correspond to another discrete eigenvalue less than or equal
        to zero.  
        It is either a discrete eigenvalue in the gap $(0, \omega -
        \tfrac{c^2}{4})$ or the base of the essential spectrum.
        Regardless, $\sigma(\widetilde{L}_{11})
        \setminus\set{-\lambda_{11}^2, 0}$ is bounded away from zero.

      \end{proof}

      Using $\chi_{11}$, we construct the negative subspace $N$.

      \begin{proposition}
        \label{prop:N}
        Let
        \begin{equation}
          \label{eqn:N_def}
          N\equiv {\rm span} \set{\chi_{-} }
        \end{equation}
        where
        \begin{subequations}
          \label{eqn:chi-}
          \begin{align}
            \chi_{-} &\equiv (\chi_{11} +i \chi_{12} ) e^{i\theta},\\
            \label{eqn:chi-_b}
            \chi_{12}& \equiv \varphi
            \bracket{-\frac{\sigma}{\sigma+1} \int_{-\infty}^y
              \varphi^{2\sigma-1}(s)\chi_{11}(s) ds +k_{12}},
          \end{align}
        \end{subequations}
        and $k_{12}\in \mathbb{R}$ is chosen such that
        \begin{equation}
          \label{eqn:chi12_orth_phi}
          \langle \chi_{12} , \varphi \rangle =0.
        \end{equation}
        For $u \in N\setminus \{ 0\}$,
        \begin{equation*}
          \langle H_\phi u, u \rangle <0.
        \end{equation*}
      \end{proposition}
      \begin{proof}
        The function $\chi_{12}$ is in $L^2$. Indeed, the integral in
        \eqref{eqn:chi-_b} is well defined since, as $|y|\rightarrow
        \infty$,
        \begin{gather*}
          \abs{\varphi(y)} \lesssim \exp\set{- \sqrt{\omega - c^2/4} \abs{y}},\\
          \abs{\chi_{11}(y)}\lesssim \exp\set{-\sqrt{\omega - c^2/4 +
              \lambda_{11}^2}\abs{y}}.
        \end{gather*}
        Thus the integrand is bounded.  

        From \eqref{eqn:H_varphi} and
        \eqref{eqn:L11_chi11},
        \begin{equation*}
          \langle H_\phi \chi_{-}, \chi_{-} \rangle =  \langle \widetilde{L}_{11} \chi_{11}, \chi_{11} \rangle =-\lambda_{11}^2 <0.
        \end{equation*}
      \end{proof}

      \subsection{The Kernel}

      In this subsection, we give an explicit characterization of the
      kernel of $H_\phi$.
      \begin{proposition}
        \label{prop:Z}
        Let
        \begin{equation}
          \label{eqn:Z_def}
          Z = {\rm span} \set{\chi_1, \chi_2}
        \end{equation}
        where
        \begin{subequations}
          \label{eqn:chi1&chi2}
          \begin{align}
            \chi_{1} &= \left(\varphi_y + i(k_2-\frac{1}{2\sigma+2}\varphi^{2\sigma} ) \varphi\right)e^{i\theta},\\
            \chi_{2} &= i\varphi e^{i\theta}
          \end{align}
        \end{subequations}
        with $k_2$ is a real constant such that
        \begin{equation}
          \label{eqn:def_k2}
          \left\langle \left(k_2 -\frac{1}{2\sigma+2} \varphi^{2\sigma} \right)\varphi, \varphi \right\rangle = 0.
        \end{equation}

        Then $Z = \ker H_{\phi}$.
      \end{proposition}
      \begin{proof}
        We first prove that $\chi_1$ and $\chi_2$ are linearly
        independent elements of the kernel, and then show that the
        kernel is at most two dimensional.  Applying $H_\phi$ (in the
        form \eqref{e:Hdephased}) to $\chi_2$ and using that $L_{21}
        \varphi =0$ and \eqref{eqn:L22=0}, we get $H_\phi \chi_2 =0$.
        For $\chi_1$, we compute
        \begin{equation*}
          \begin{split}
            &L_{11} \varphi_y + L_{21}
            (k_2-\frac{1}{2\sigma+2}\varphi^{2\sigma}
            ) \varphi\\
            & \quad = \widetilde{L}_{11}\varphi_y +
            \frac{\sigma^2}{(\sigma+1)^2}\varphi^{4\sigma} \varphi_y +
            k_2 L_{21}
            \varphi - \frac{1}{2\sigma +2}L_{21} \varphi^{2\sigma+1}\\
            &\quad = \frac{\sigma^2}{(\sigma+1)^2}\varphi^{4\sigma}
            \varphi_y - \frac{1}{2\sigma
              +2}\frac{2\sigma^2}{\sigma+1}\varphi^{4\sigma}\varphi_y=0
          \end{split}
        \end{equation*}
        and
        \begin{equation*}
          \begin{split}
            &L_{12} \varphi_y + L_{22}
            (k_2-\frac{1}{2\sigma+2}\varphi^{2\sigma}
            ) \varphi\\
            & \quad = -\frac{2\sigma^2+\sigma}{\sigma+1}
            \varphi^{2\sigma-1}\varphi_y^2 -
            \frac{\sigma}{\sigma+1}\varphi^{2\sigma} \varphi_{yy} \\
            &\quad \quad - \frac{1}{2\sigma +2}\bracket{-2\sigma
              (2\sigma+1) \varphi^{2\sigma-1} \varphi_y^2 - 2\sigma
              \varphi^{2\sigma} \varphi_{yy}+ \varphi^{2\sigma} L_{22}
              \varphi }=0.
          \end{split}
        \end{equation*}
        Thus, $Z \subset \ker H_\phi$, and the kernel is at least two
        dimensional.

        We now show that it is exactly two dimensional.  If we
        consider the problem
        \[
        H_\phi f =0,
        \]
        as a second order system of two real valued functions, we know
        there are four linearly independent solutions.  As $y \to
        -\infty$, two of these solutions decay exponentially, while
        two grow exponentially.  Thus, there are at most two linearly
        independent solutions which are spatially localized.  Hence, $Z
        = \ker H_\phi$.

      \end{proof}

      \subsection{The Positive Subspace and Proof of the Spectral
        Decomposition}

      We define the subspace $P$ and prove Theorem
      \ref{thm:Spectral_H}.  For that, we need the following lemmas
      about $\widetilde{L}_{11}$ and $L_{22}$.
      \begin{lemma}
        \label{l:L11}
        For any real function $f\in H^1(\mathbb{R})$ satisfying the
        orthogonality conditions
        \begin{equation}
          \label{eqn:f_x11}
          \langle f, \varphi_y \rangle = \langle f, \chi_{11} \rangle =0,
        \end{equation}
        there exists a positive number $\delta_{11}>0$, such that
        \begin{equation}
          \langle \widetilde{L}_{11} f, f \rangle \geq \delta_{11} \| f\|^2_{H^1}.
        \end{equation}
      \end{lemma}

      \begin{proof}

        From Lemma \ref{l:spec_L11},
        \eqref{e:L11bnd} holds on the subspace orthogonal to
        $\varphi_y$ and $\chi_{11}$, so
        \[
        \langle \widetilde{L}_{11} f, f \rangle \geq \mu_{11} \|
        f\|^2_{L_2}.
        \]
        To get the $H^1$ lower bound, let
        \[
        V_1(y) = \omega -\frac{c^2}{4} +
        \frac{c(2\sigma+1)}{2}\varphi^{2\sigma}
        -\frac{8\sigma^2+6\sigma+1}{4(\sigma+1)^2}\varphi^{4\sigma},
        \]
        so that $\widetilde{L}_{11} = -\partial_{yy} + V_1$, with $\|
        V_1 \|_{L^{\infty}}<\infty.$ Thus,
        \begin{equation*}
          \begin{split}
            \langle \widetilde{L}_{11} f, f \rangle &= \langle -\partial_{yy} f, f \rangle + \langle V_1 f, f \rangle \\
            &\geq \langle -\partial_{yy} f, f \rangle - \|V_1\|_{L^\infty} \|f\|_{L^2}^2 \\
            &= \|\partial_y f\|_{L^2}^2 -
            \frac{1}{\mu_{11}}\|V_1\|_{L^\infty}\langle
            \widetilde{L}_{11} f, f \rangle.
          \end{split}
        \end{equation*}
        It follows that
        \begin{equation*}
          \langle \widetilde{L}_{11} f, f \rangle \geq \frac{1}{1+\mu_{11}^{-1} \|V_1\|_{L^\infty}} \|\partial_y f\|_{L^2}^2
        \end{equation*}
        Taking $\delta_{11}$ sufficiently small, we have
        \begin{equation*}
          \langle \widetilde{L}_{11} f, f \rangle \geq \delta_{11} \| f\|^2_{H^1}.
        \end{equation*}
      \end{proof}

      \begin{lemma}
        \label{l:L22}
        For any real function $f\in H^1(\mathbb{R})$ satisfying
        \begin{equation}
          \langle f, \varphi \rangle =0,
        \end{equation}
        there exists a positive number $\delta_{22}>0$, such that
        \begin{equation}
          \label{eqn:L22_delta22}
          \langle L_{22} f, f \rangle \geq \delta_{22} \| f\|^2_{H^1}.
        \end{equation}
      \end{lemma}

\begin{proof}
  As was the case for $\widetilde{L}_{11}$, $L_{22}$ is a relatively
  compact perturbation of $-\partial_y^2 + \omega - c^2/4$, so it also
  has
  \begin{equation*}
    \sigma_{\rm ess}(L_{22}) = \left [\omega-\tfrac{c^2}{4},\infty\right).
  \end{equation*}
  Thus, all points in the spectrum below $\omega -c^2/4$ correspond to
  discrete eigenvalues.  From \eqref{eqn:L22=0} and $\varphi$ is
  strictly positive, Sturm-Liouville theory tells us that zero is the
  lowest eigenvalue.  Let
  \begin{equation}
    \mu_{22} \equiv \inf_{f\neq 0, f\perp \varphi} \frac{\inner{L_{22} f}{f}}{\inner{f}{f}}.
  \end{equation}
  We know that $\mu_{22} >0$, otherwise this would contradict with $0$
  being the lowest eigenvalue.  Therefore
  \[
  \inner{L_{22} f}{f} \geq \mu_{22} \norm{f}_{L^2}.
  \]
  Using the same argument as in Lemma \ref{l:L11}, we obtain
  \eqref{eqn:L22_delta22}.
\end{proof}
We now prove Theorem \ref{thm:Spectral_H}.
\begin{proof}
  Recall $N, Z$ as defined define by \eqref{eqn:N_def} and
  \eqref{eqn:Z_def}.  We define $P$ as
  \begin{equation}
    \label{eqn:P_def}
    \begin{split}
      P=\{ p \in X\mid \langle \Re( e^{-i\theta}p), \chi_{11} \rangle
      = \langle \Re( e^{-i\theta}p), \varphi_y \rangle =\langle \Im
      (e^{-i\theta}p), \varphi \rangle =0\}.
    \end{split}
  \end{equation}
  We express $u\in X$ as
  \begin{equation*}
    u=a_1 \chi_{-} +(b_1\chi_1 +b_2 \chi_2)+p,
  \end{equation*}
  where
  \[
  a_1 = \langle u_1, \chi_{11} \rangle,\quad b_1 = \frac{\langle u_1,
    \varphi_y \rangle}{\| \varphi_y \|^2_{L^2}},\quad b_2 =
  \frac{\langle u_2, \varphi \rangle}{\| \varphi \|_{L^2}^2},
  \]
  with $u_1$ and $u_2$ are real and imaginary part of $e^{-i\theta}u$.
  Clearly, $a_1 \chi_{-} \in N$ and $b_1\chi_1 + b_2 \chi_2 \in Z$. It
  suffices to show $p \in P$. We write $p=(p_1 +i p_2)e^{i\theta}$  with $p_1$ and $p_2$ real.  Since $\varphi_y$ is odd and $\chi_{11}$
  is even,
  $ \langle \varphi_y, \chi_{11} \rangle =0,$ and we readily check
  that $\langle p_1 , \chi_{11} \rangle= \langle p_1, \varphi_y
  \rangle =0$.
  Furthermore, by \eqref{eqn:chi12_orth_phi} and \eqref{eqn:def_k2},  we also have $\langle p_2, \varphi \rangle =0$.  Thus, $p \in P$,
  and $u$ is indeed decomposed into elements of $N$, $Z$ and $P$.

  Finally, we show that $H_\phi$ is positive on $P$.  Let $\tilde p_2
  = \varphi^{-1} p_2$. By \eqref{eqn:H_varphi},
  \begin{equation}
    \label{eqn:H_p_2}
    \langle H_\phi p, p \rangle
    = \langle \widetilde{L}_{11} p_1, p_1 \rangle + \int_{-\infty}^\infty (\varphi \partial_y \tilde p_2  + \frac{\sigma}{\sigma+1} \varphi^{2\sigma} p_1)^2 dy.
  \end{equation}
  Lemma \ref{l:L11} gives the desired lower bound on the first
  term. For the second term, we break it into two cases, depending on
  how $\norm{\varphi \partial_y \tilde p_2}_{L^2}$ and
  $\norm{p_1}_{L^2}$ compare.  Let
  \begin{equation}
    C_\sigma \equiv \frac{2\sigma}{\sigma +1}\norm{\varphi}_{L^\infty}^{2\sigma}
  \end{equation}
  \begin{enumerate}[(a)]
  \item If $\| \varphi \partial_y \tilde p_2 \|_{L^2} \geq C_\sigma\|
    p_1 \|_{L^2}$, we estimate the second term in \eqref{eqn:H_p_2} as
    follows,
    \begin{equation*}
      \begin{split}
        \norm{\varphi \partial_y \tilde p_2 + \frac{\sigma}{\sigma+1}
          \varphi^{2\sigma} p_1}_{L^2}&
        \geq \norm{\varphi \partial_y \tilde p_2}_{L^2} -
        \frac{\sigma}{\sigma +1}\norm{\varphi}_{L^\infty}^{2\sigma}
        \norm{p_1}_{L^2}\\
        & = \norm{\varphi \partial_y \tilde p_2}_{L^2} -
        \frac{1}{2}C_{\sigma} \norm{p_1}_{L^2} \geq
        \frac{1}{2}\norm{\varphi \partial_y \tilde p_2}_{L^2}
      \end{split}
    \end{equation*}
    By \eqref{eqn:B2bb}, we then have
    \[
    \norm{\varphi \partial_y \tilde p_2 + \frac{\sigma}{\sigma+1}
      \varphi^{2\sigma} p_1}_{L^2}^2\geq
    \frac{1}{4}\inner{L_{22}p_2}{p_2}
    \]
    By Lemmas \ref{l:L11} and \ref{l:L22}, we get
    \begin{equation*}
      \begin{split}
        \inner{H_\phi p}{p} \geq \inner{\widetilde{L}_{11} p_1}{p_1} +
        \frac{1}{4}\inner{L_{22}p_2}{p_2}
        \geq \delta_{a} \norm{p}_{H^1}^2,
      \end{split}
    \end{equation*}
    for some small enough $\delta_a$.

  \item If instead, $\| \varphi \partial_y \tilde p_2 \|_{L^2} <
    C_\sigma \| p_1 \|_{L^2},$ then,
    \begin{equation}
      \begin{split}
        \langle H_\phi p, p \rangle &\geq \inner{\widetilde{L}_{11}
          p_1}{p_1}
        \geq \frac{\delta_{11}}{2} \| p_1 \|_{H^1}^2 +\frac{\delta_{11}}{2} \| p_1 \|_{L^2}^2\\
        & \geq \frac{\delta_{11}}{2} \| p_1 \|_{H^1}^2 + \frac{\delta_{11}}{2C_\sigma^2} \| \varphi \partial_y \tilde p_2 \|_{L^2}^2\\
        & = \frac{\delta_{11}}{2} \| p_1 \|_{H^1}^2 +
        \frac{\delta_{11}}{2C_\sigma^2} \langle L_{22} p_2, p_2
        \rangle
        \geq \delta_b \norm{p}_{H^1}^2.
      \end{split}
    \end{equation}
  \end{enumerate}
  Taking the smaller value of $\delta_a$ and $\delta_b$ as $\delta$,
  we have
  \begin{equation}
    \begin{split}
      \langle H_\phi p, p \rangle \geq \delta \|p \|_{H^1}^2.
    \end{split}
  \end{equation}
  It follows that $N$, $Z$ and $P$ have trivial intersection amongst
  one another.  Hence $X = N+Z+P.$
\end{proof}

\section{Analysis of the Hessian Matrix}
\label{sec:concavity_d}
In this section, we compute the number of the positive eigenvalues of
the Hessian matrix of $d(\omega,c)$, $p(d''(\omega,c))$.  Since the
number of negative eigenvalues of $H_{\phi_{\omega,c}}$ is in all
cases equal to one, $p(d''(\omega,c))$ will determine whether or not
the soliton is stable.

To make this assessment, we examine the determinant and the trace of
$d''(\omega,c)$.  From Lemmas \ref{lemma:d} and \ref{le:QcQw_PcPw},
the determinant can be expressed as
\begin{equation}
  \label{e:detHessd}
  \begin{split}
    \det[d''(\omega,c)]&=\partial_\omega Q \partial_c P -\partial_c Q \partial_\omega P\\
    &=
    2^{-\frac{2}{\sigma}-4}\sigma^{-2}(1+\sigma)^{\frac{2}{\sigma}}(4\omega-c^2)^{\frac{2}{\sigma}-1}  \omega^{-\frac{1}{\sigma}-2}\\
    &\quad\times\left[4(\sigma-1)\omega \alpha_0  -2\sqrt{\omega}c\alpha_0 +(4\omega-c^2)\alpha_1\right]\\
    &\quad\times\left[4(\sigma-1)\omega \alpha_0
      +2\sqrt{\omega}c\alpha_0 -(4\omega-c^2)\alpha_1\right],
  \end{split}
\end{equation}
where
\begin{align*}
  \alpha_n(\omega,c;\sigma)&\equiv \int_0^\infty h
  ^{-\frac{1}{\sigma}-n} dx >0,\\
  h(x;\sigma;\omega,c) &\equiv \cosh(\sigma\sqrt{4\omega - c^2}x) -
  \tfrac{c}{2\sqrt{\omega}}.
\end{align*}

Meanwhile, the trace is
\begin{equation}
  \label{e:trHessd}
  \begin{split}
    \tr[d''(\omega,c)]&=\partial_\omega Q+\partial_c P\\
    &=2^{-\frac{1}{\sigma}-2} \sigma^{-1}
    (1+\sigma)^{\frac{1}{\sigma}} (4\omega-c^2)^{\frac{1}{\sigma}-1}
    (1+\omega)
    \omega^{-\frac{1}{2\sigma}-\frac{3}{2}}\\
    &\quad\quad\times(c(c^2-4\omega)\alpha_1+2\sqrt{\omega}(c^2-4(\sigma-1)\omega)\alpha_0).
  \end{split}
\end{equation}

\begin{theorem}
  \label{thm:Hessian_d''}
  If $\sigma\geq2$, and $4\omega>c^2$, $p(d''(\omega,c))=0$.
\end{theorem}
\begin{proof}
  We examine the terms appearing in \eqref{e:detHessd}.  The first
  term is clearly positive.  The second term is also positive,
  \begin{equation*}
    \begin{split}
      &4(\sigma-1)\omega \alpha_0 -2\sqrt{\omega}c\alpha_0 +(4\omega-c^2)\alpha_1 \\
      =&4\omega\left[(\sigma-1-\frac{c}{2\sqrt{\omega}})\alpha_0+(1-\frac{c^2}{4\omega})\alpha_1\right]>0.
    \end{split}
  \end{equation*}
  For the third term,
  \begin{equation*}
    \begin{split}
      &4(\sigma-1)\omega \alpha_0
      +2\sqrt{\omega}c\alpha_0 -(4\omega-c^2)\alpha_1\\
      \geq&(4\omega+2\sqrt{\omega}c)\alpha_0 -(4\omega-c^2)\alpha_1\\
      =&4\omega(1+\frac{c}{2\sqrt{\omega}})\int_0^\infty
      h^{-\frac{1}{\sigma}-1}(\cosh(\sigma\sqrt{4\omega-c^2}x)-1)
      dx>0.
    \end{split}
  \end{equation*}
  Thus $ \det[d''(\omega,c)]>0$, implying the eigenvalues of
  $d''(\omega,c)$ have the same sign.  Turning to the trace,
  $c^2-4(\sigma-1)\omega\leq c^2-4\omega<0$ for $\sigma\geq 2$.  By
  \eqref{e:trHessd}, $ \tr[d''(\omega,c)]<0.$ Hence, the two
  eigenvalues of $d''(\omega,c)$ are negative.
\end{proof}

Closely related to $\det[d'']$ is the function
\begin{equation}
  \label{e:Fsigz}
  \begin{split}
    F(z;\sigma) \equiv & (\sigma-1)^2 \bracket{\int_0^\infty (\cosh y -z)^{-\frac{1}{\sigma}} dy}^2\\
    &- \bracket{\int_0^\infty (\cosh y -z)^{-\frac{1}{\sigma}-1}(z
      \cosh y -1)dy }^2,
  \end{split}
\end{equation}
which helps count the number of positive and negative eigenvalues for
$\sigma \in (0,2)$.  Indeed,
\begin{lemma}
  \label{thm:F_and_det}
  For $\sigma\in (0,2)$ and admissible $(\omega,c)$,
  $\det[d''(\omega,c)]$ has the same sign as
  $F(\frac{c}{2\sqrt{\omega}};\sigma)$.
\end{lemma}
\begin{proof} We rewrite \eqref{e:detHessd} as,
  \begin{equation*}
    \small
    \begin{split}
      &\frac{\det[d''(\omega,c)]} {2^{-\frac{2}{\sigma}-4} \sigma^{-2}(1+\sigma)^{\frac{2}{\sigma}} (4\omega-c^2)^{\frac{2}{\sigma}-1} \omega^{-\frac{1}{\sigma}-2} } \\
      =&16(\sigma-1)^2 \omega^2 \alpha_0^2 -
      \bracket{\alpha_1(c^2-4\omega)
        +2\sqrt{\omega}c\alpha_0}^2\\
      =&16 \omega^2 \set{(\sigma-1)^2 \alpha_0^2 -
        \bracket{\int_0^\infty
          h^{-\frac{1}{\sigma}-1}\paren{\tfrac{c}{2\sqrt{\omega}}
            \cosh(\sigma\sqrt{4\omega-c^2}x) -1}dx }^2 }.
    \end{split}
  \end{equation*}
  Letting $y = \sigma\sqrt{4\omega-c^2} x$,
  \begin{equation}
    \label{eqn:d&F}
    \begin{split}
      \frac{\det[d''(\omega,c)]}{2^{-\frac{2}{\sigma}-4}\sigma^{-2}(1+\sigma)^{\frac{2}{\sigma}}(4\omega-c^2)^{\frac{2}{\sigma}-1}
        \omega^{-\frac{1}{\sigma}-2}}
      = \frac{16 \omega^2}{\sigma^2(4\omega-c^2)}
      F\left (\frac{c}{2\sqrt{\omega}};\sigma\right).
    \end{split}
  \end{equation}
\end{proof}

When $\sigma=1$ and $z\in (-1,1)$, we have $F(z;\sigma)=-1$ and
$\det[d''(\omega,c)]=-1/\omega$.  For $\sigma \in (0,1)\cup(1,2)$,
we can evaluate the function $F(z;\sigma)$ numerically, as shown in Figures
\ref{f:Fcurves} and \ref{f:Fcurves01}. For any fixed $\sigma \in
(1,2)$, $F(z;\sigma)$ is monotonically increasing in $z$ and has
exactly one root $z_0$ in the interval $(-1,1)$.  For fixed
$\sigma\in(0,1)$, $F(z;\sigma)$ is monotonically decreasing in $z$ and
strictly negative.  It is this numerical computation of $F$ which is
used to complete the proofs of Theorems \ref{thm:main2} and
\ref{thm:main3}.  In contrast, for $\sigma\geq 2$, we can prove that
$F(z;\sigma)$ is strictly positive without resorting to computation.

\begin{figure}
  \subfigure[]{\includegraphics[width=10cm]{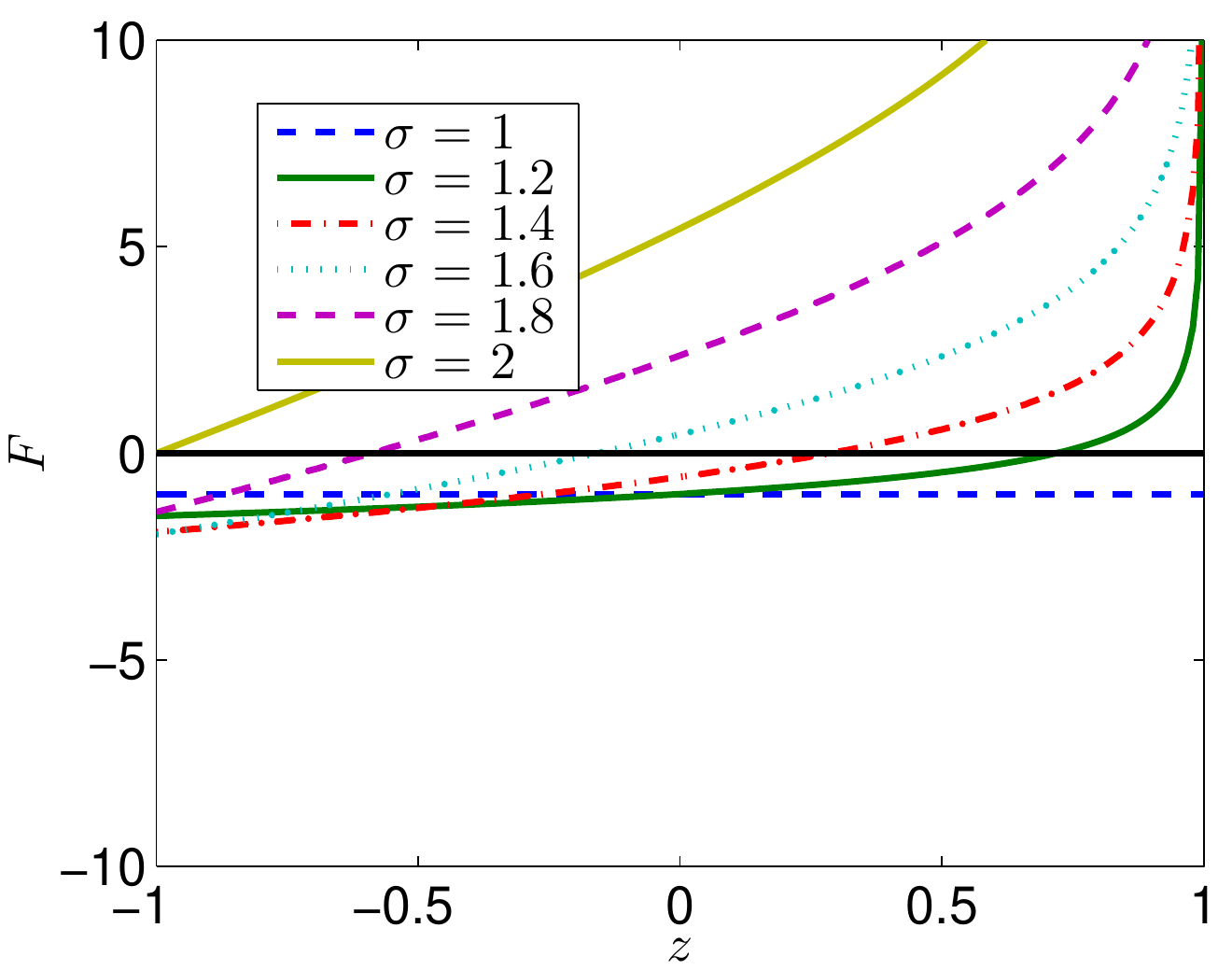}}
  \subfigure[]{\includegraphics[width=10cm]{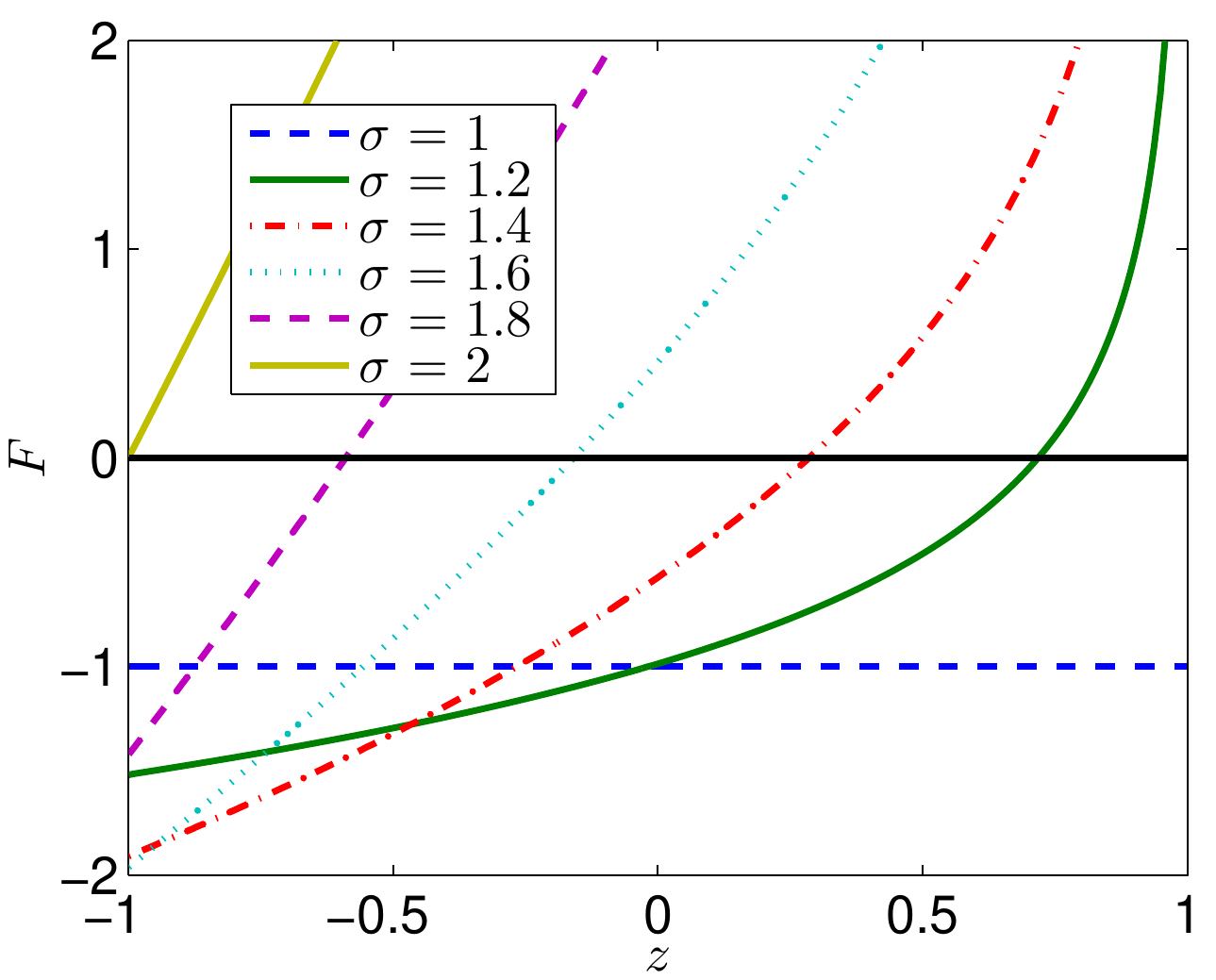}}
  \caption{(a) Function $F(z;\sigma)$, \eqref{e:Fsigz}, for several
    values of $\sigma\in [1,2]$. (b) is a magnified plot near the
    $z$-axis.}
  \label{f:Fcurves}
\end{figure}

\begin{figure}
  \subfigure[]{\includegraphics[width=10cm]{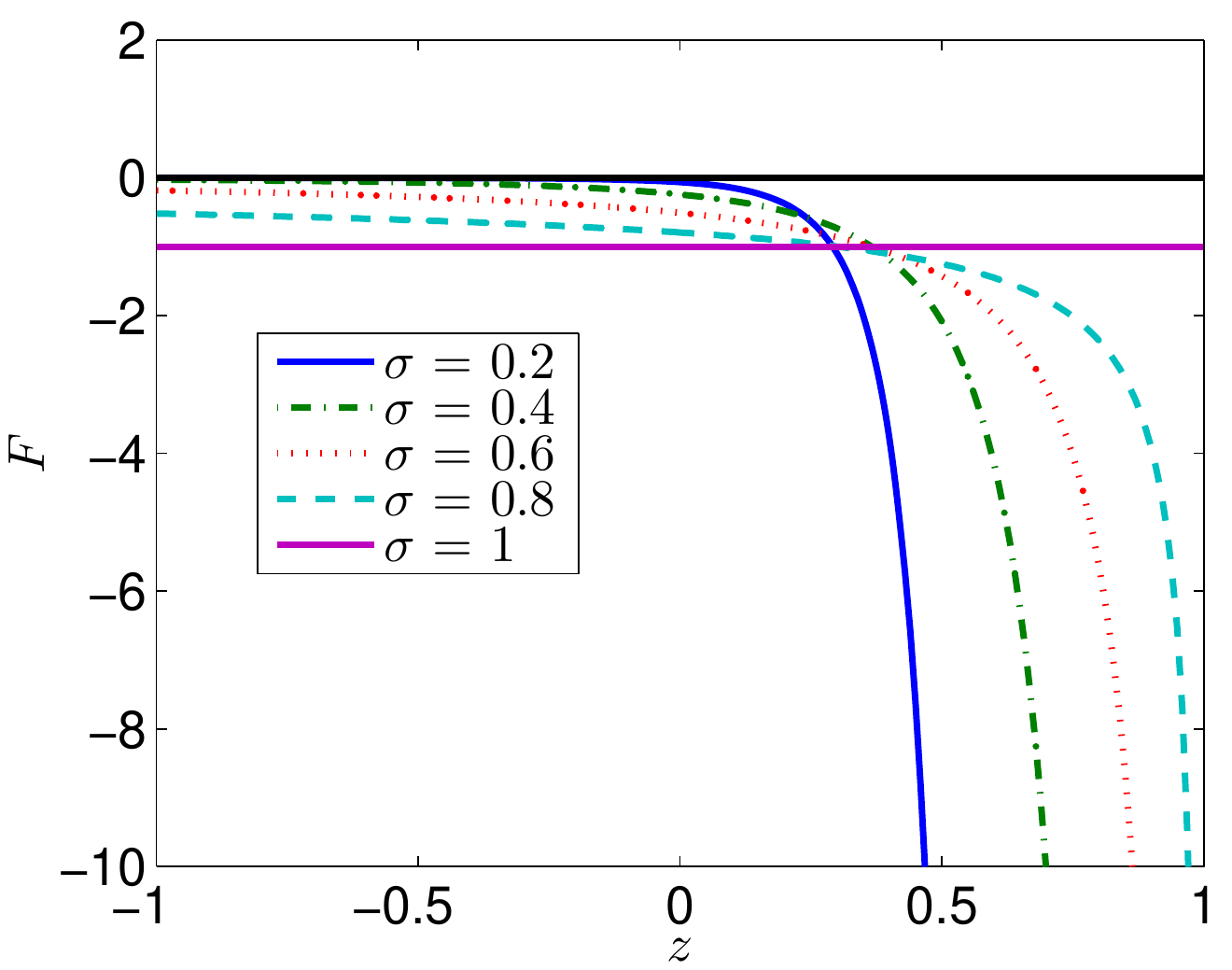}}
  \subfigure[]{\includegraphics[width=10cm]{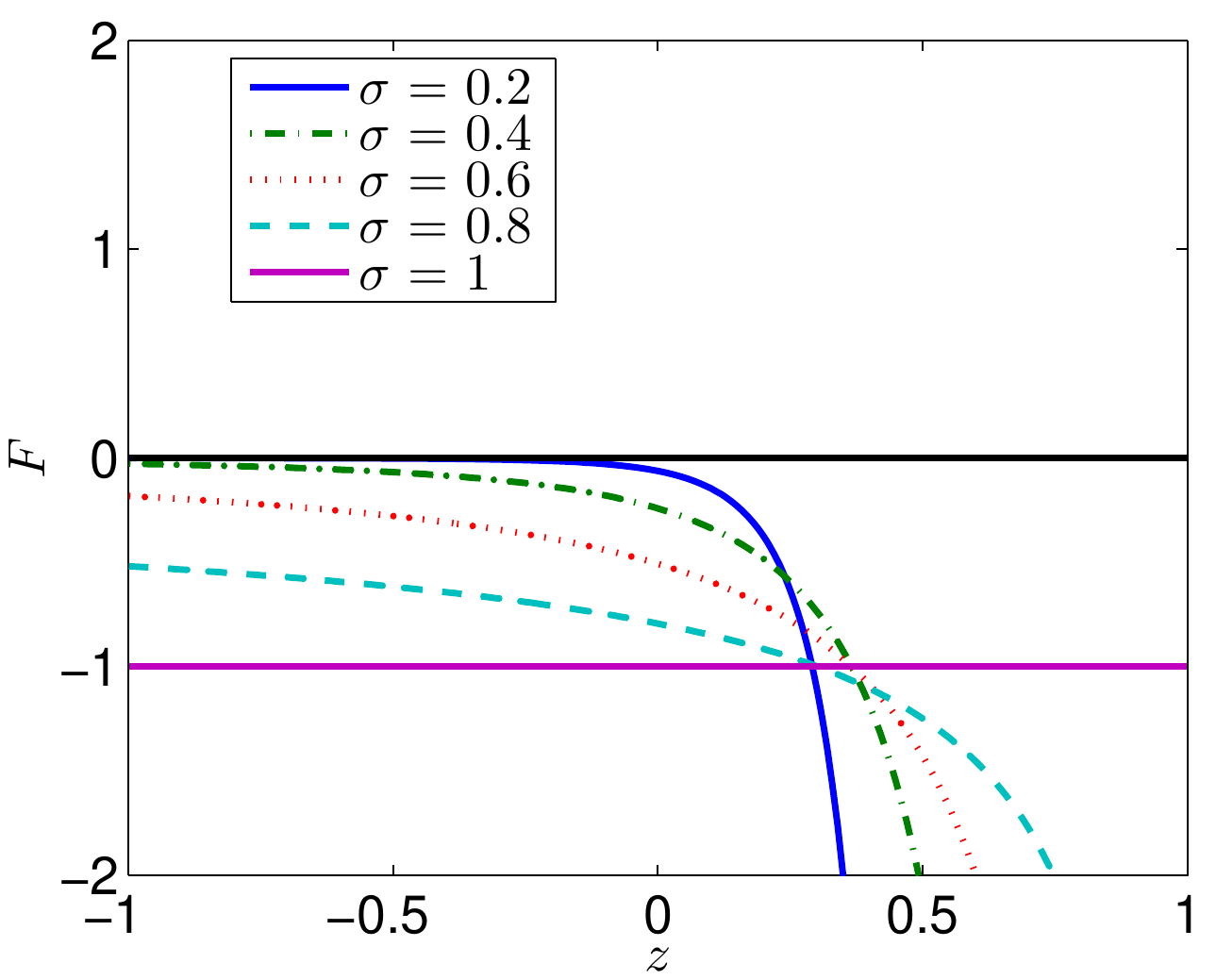}}
  \caption{(a) Function $F(z;\sigma)$, \eqref{e:Fsigz}, for several
    values of $\sigma\in (0,1]$. (b) is a magnified plot near the
    $z$-axis.}
  \label{f:Fcurves01}
\end{figure}

We thus have the following theorem about $p(d''(\omega,c))$:
\begin{theorem}[Numerical]
  \label{thm:Hessian_d''_12}
  For admissible $(\omega,c)$,
  \begin{enumerate}[(i)]
  \item when $\sigma\in(1,2)$ and $ c = 2 z_0 \sqrt{\omega}$,
    $\det[d''(\omega,c)]=0$,
  \item when $\sigma\in(1,2)$ and $ c < 2 z_0 \sqrt{\omega}$,
    $\det[d''(\omega,c)]<0$; $p(d''(\omega,c))=1$,
  \item when $\sigma\in(1,2)$ and $ c > 2 z_0 \sqrt{\omega}$,
    $\det[d''(\omega,c)]>0$; $p(d''(\omega,c))=0$ or
    $p(d''(\omega,c))=2$.
  \item when $\sigma\in(0,1)$,$\det[d''(\omega,c)]<0$;
    $p(d''(\omega,c))=1$.
  \item when $\sigma =1$, $\det[d''(\omega,c)]=-1/\omega<0$;
    $p(d''(\omega,c))=1$.
  \end{enumerate}
\end{theorem}

\section{Orbital Stability and Instability}
\label{sec:Orbital_stabi}
In this section, we complete the stability/instability proofs.
\begin{proof}[Proof of Theorem \ref{thm:main}]
  From Theorem \ref{thm:Spectral_H} and Theorem \ref{thm:Hessian_d''},
  $n(H)=1$ for any $\sigma>0$, and $p(d'')=0$ for $\sigma \geq 2$.
  Thus $n(H)-p(d'')=1,$ is odd and all solitary waves are orbitally
  unstable by Theorem \ref{thm:GSS}.
\end{proof}

\begin{proof}[Proof of Theorem \ref{thm:main2}]
  By assumption, for each $\sigma \in (1,2)$, there exists $z_0$, a
  unique zero crossing of \eqref{e:Fsigz}, and this function is
  monotonically increasing.  By Theorem \ref{thm:Hessian_d''_12},
  $\det[d''(\omega,c)]<0$ for admissible $(\omega,c)$ satisfying $c<2
  z_0\sqrt{\omega} $. It follows that there is one positive and one
  negative eigenvalues for $d''(\omega,c)$. Therefore,
  $p(d'')=1$. Furthermore, from the theorem \ref{thm:Spectral_H},
  $n(H)=1$. Hence
  $$n(H)-p(d'')=0,$$
  and we have the orbital stability of solitary waves.

  When $c>2 z_0\sqrt{\omega} $, also by Theorem
  \ref{thm:Hessian_d''_12}, $\det[d''(\omega,c)]>0$. So the signs of
  the two eigenvalues of $d''(\omega,c)$ are the same. If both of the
  eigenvalues were positive, then $2 = p(d'') > n(H_\phi) =1$. This
  contradicts \eqref{eqn:p(d)}. Hence both of the eigenvalues are
  negative and $p(d'')=0$. It follows that
  $$n(H) -p(d'')=1,$$
  and we have the orbital instability of solitary waves.
\end{proof}

Following the same argument, we can prove Theorem \ref{thm:main3}.
\begin{proof}[Proof of Theorem \ref{thm:main3}]
  When $\sigma\in (0,1]$, from Theorem \ref{thm:Hessian_d''_12},
  $\det[d''(\omega,c)]<0$ for admissible $(\omega,c)$. Consequently,
  $d''(\omega,c)$ has one positive eigenvalue and one negative
  eigenvalue; $p(d'')=1$. By Theorem \ref{thm:Spectral_H},
  $n(H)=1$. Hence
  $$n(H)-p(d'')=0,$$
  and the solitary waves are orbital stable.
\end{proof}

\section{Discussion and Numerical Illustration}
\label{s:disc}

We have explored the stability and instability of solitons for a
generalized derivative nonlinear Schr\"odinger equation.  We have
found that for $\sigma \geq 2$, all solitons are orbital unstable.
Using a numerical calculation of the function $F(z;\sigma)$ defined in
\eqref{e:Fsigz}. we have also shown that for $0< \sigma \leq 1$, all
solitons are orbital stable.  For $1 < \sigma <2$, our computation of
$F(z;\sigma)$ indicates there exist both stable and unstable solitons,
depending on the values of $\omega$ and $c$. In particular, for fixed
$\omega$ and $\sigma>1$, there are always both stable and unstable
solitons for properly selected $c$.  For $\sigma$ near 1, the unstable
solitons are always rightward moving, but, as Figure \ref{f:Fcurves}
shows, the root, $z_0$, becomes negative as $\sigma$ approaches 2.
Once $z_0 < 0$, unstable solitons can be both rightward and leftward
moving.

Other dispersive PDEs possessing both stable and unstable solitons, such as NLS and KdV with
saturating nonlinearities, \cite{Sulem1999,
  Comech:2003ur,Comech:2007uk,Marzuola:2010wc,Ohta:2011dh}, achieve
this by introducing a nonlinearity that breaks scaling.  In contrast,
gDNLS always has a scaling symmetry, and throughout the regime $1 <
\sigma <2$, the scaling is $L^2$ supercritical.  This also implies the existence of an entire manifold of {\it
  critical} solitons, precisely when $c =2 z_0 \sqrt{\omega}$.  Along
this curve, the standard stability results of
\cite{Grillakis1990,Grillakis1987,Weinstein1985,Weinstein1986a}, break
down, and a more detailed analysis is required.  In
\cite{Grillakis1990}, the stability can be demonstrated in this
degenerate case provided $d(\omega,c)$ remains convex.  Given that
within any neighborhood of a critical soliton there exist unstable
solitons, we conjecture that it is unstable.  While there has been
recent work on critical one parameter solitons for NLS type equations,
\cite{Comech:2003ur,Comech:2007uk,Marzuola:2010wc,Ohta:2011dh}, to the
best our knowledge, there has not been an analogous work on two
parameter solitons.

While the equation retains the scaling symmetry, we observe that, in
contrast to NLS solitons, not all gDNLS solitons can be obtained from
scaling.  Indeed, for \eqref{e:NLS}, all solitons $e^{i\lambda t}
R({\bf x};\lambda)$, solving
\[
-\Delta R + \lambda R - \abs{R}^{2\sigma}R=0
\]
can be obtained from the $\lambda=1$ soliton via the transformation
\[
e^{i\lambda t} R({\bf x}; \lambda) = e^{i \lambda t}
\lambda^{\frac{1}{2\sigma}} R(\lambda^{\frac{1}{2}}{\bf x};1).
\]
In contrast, while the gDNLS solitons also inherit the scaling
symmetry of gDNLS, not all admissible $(\omega,c)$ can be scaled to a
particular soliton.  Instead,
\begin{equation}
  \psi_{\omega,c}(x,t) = e^{i\omega t}\phi_{\omega,c}(x-ct) = e^{i
    \omega t} \phi_{1, c/\sqrt{\omega}}(\sqrt{\omega}(x-ct)).
\end{equation}
Only solitons for which
\[
\frac{c_1}{\sqrt{\omega_1}} = \frac{c_2}{\sqrt{\omega_2}}
\]
can be scaled into one another.

Our results were based on the assumption that a weak solution existed.
While we do not have an $H^1$ theory in general, our results can, in
part, be made rigorous as follows.  For $\sigma \geq 2$, one should be
able to apply the technique of \cite{Tsutsumi1980} to obtain a local
solution in $H^s$, with $s>1$.  Alternatively, for $\sigma \geq 2$ and
integer valued, \cite{Kenig:1993wr,Kenig:1998wr} can be invoked.
Again, this yields a local solution in $H^s$, $s>1$.  For $s$
sufficiently large, the solution will also conserve the invariants.

This is sufficient to fully justify the instability of the unstable
solitons, since there is sufficient regularity such that if the
solution leaves a neighborhood of the soliton in $H^1$, it also leaves
in $H^s$, $s>1$.  However, this is insufficient to prove stability,
because even if the solution stays close in the $H^1$ norm, the norm
of the solution could grow in a higher Sobolev norm.

There is also the question of the monotonicity of $F$, for which we
relied on numerical computation for $\sigma <2$.  Looking at Figures
\ref{f:Fcurves} and \ref{f:Fcurves01}, it would appear that the
$F(z;\sigma =2)$ is an upper bound on $F(z;\sigma)$ for $1<\sigma <2$.
In addition, there appears to be a singularity at $z=1$.  Likewise,
the line $F=0$ appears to be an upper bound in the range $\sigma \leq
1$.  A more subtle analysis may permit a rigorous justification of our
work in this regime.

Lastly, we provide some numerical experiments of solitons on both the
stable and unstable branches.  We studied the stability near the
turning point $ c = 2 z_0 \sqrt{\omega}$. When $\sigma=1.5$, $z_0 =
0.0618303$. The initial condition is chosen as
\begin{equation}
  \label{e:sim_ic}
  \psi_0(x,0)=\psi_{\omega,c}(x,0)+ 0.0001 e^{-2
    x^2}.
\end{equation}
We simulate \eqref{eqn:gDNLS} using the fourth order exponential time
difference scheme of\cite{Kassam2005}, and treat the nonlinearity
pseudospectrally.  Though the nonlinearity is not polynomial in its
arguments, $\psi$, $\bar \psi$ and $\psi_x$,
\[
\abs{\psi}^{3} \psi_x = \psi \bar \psi \abs{\psi}\psi_x,
\]
we found that dealiasing as though it were a quintic problem proved
robust.

Our results are as follows:
\begin{enumerate}
\item When $\omega=1$ and $c =0 < 2z =0.1236606$, Figure
  \ref{fig:orb_stable} shows that the solitary wave retains its shape
  for a long time $(t=100)$.
\item When $\omega=1$ and $c =0.2 > 2z =0.1236606$, Figure
  \ref{fig:orb_instable} shows that the amplitude of the solitary wave
  increases rapidly near $t=10$ and it is not orbitally stable.
\end{enumerate}
Our simulation of the unstable soliton suggests that, rather than
disperse or converge to a stable soliton, gDNLS may result in a finite
time singularity.  We will explore the potential for singularity
formation in the forthcoming work \cite{Liu:fk}.

\begin{figure}
  \includegraphics[width=0.9\linewidth]{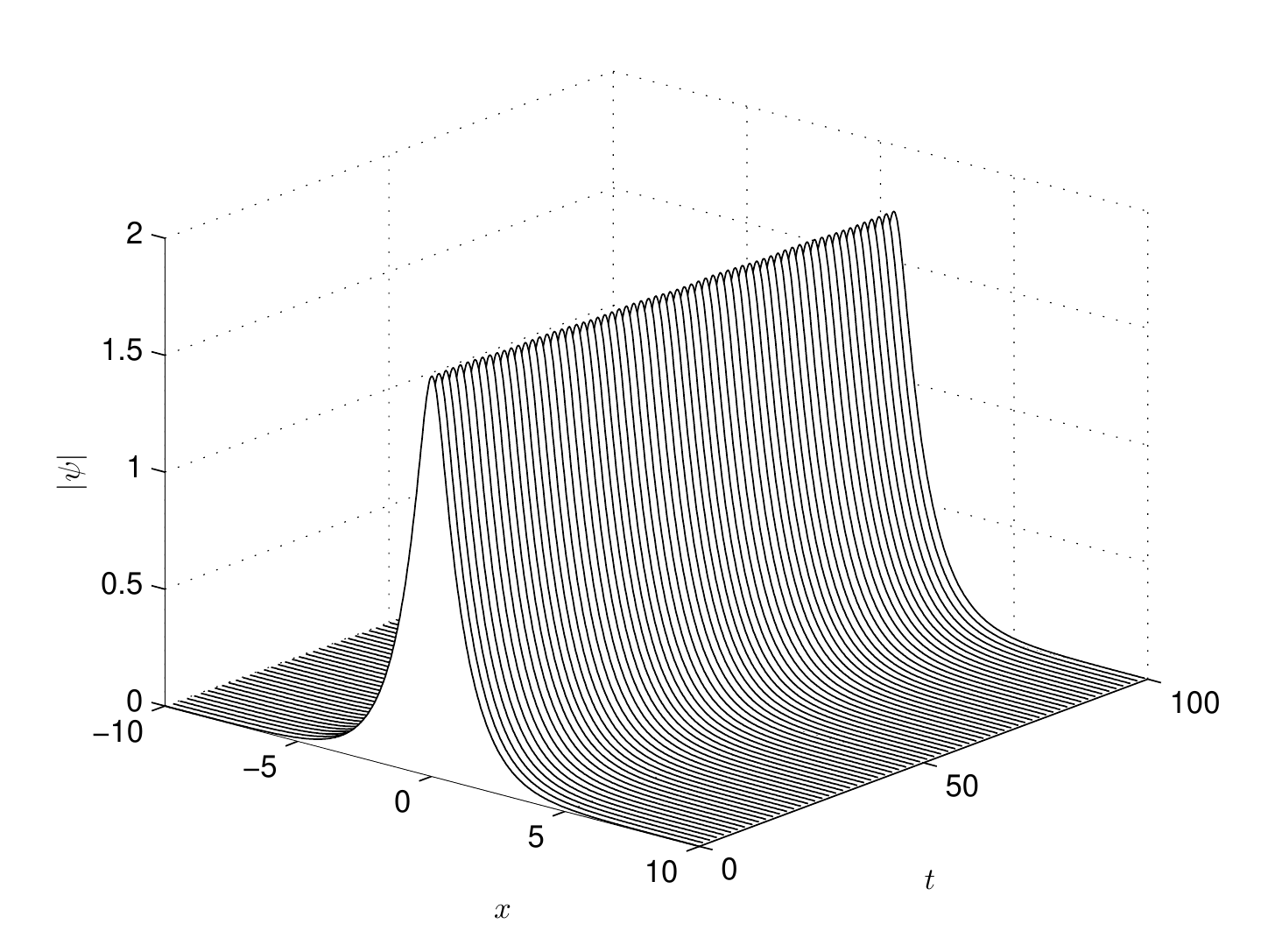}\\
  \caption{Evolution of a perturbed orbitally stable soliton, $\omega
    =1$ and $c=0$, with initial condition
    \eqref{e:sim_ic}. }\label{fig:orb_stable}
\end{figure}

\begin{figure}
  \includegraphics[width=0.9\linewidth]{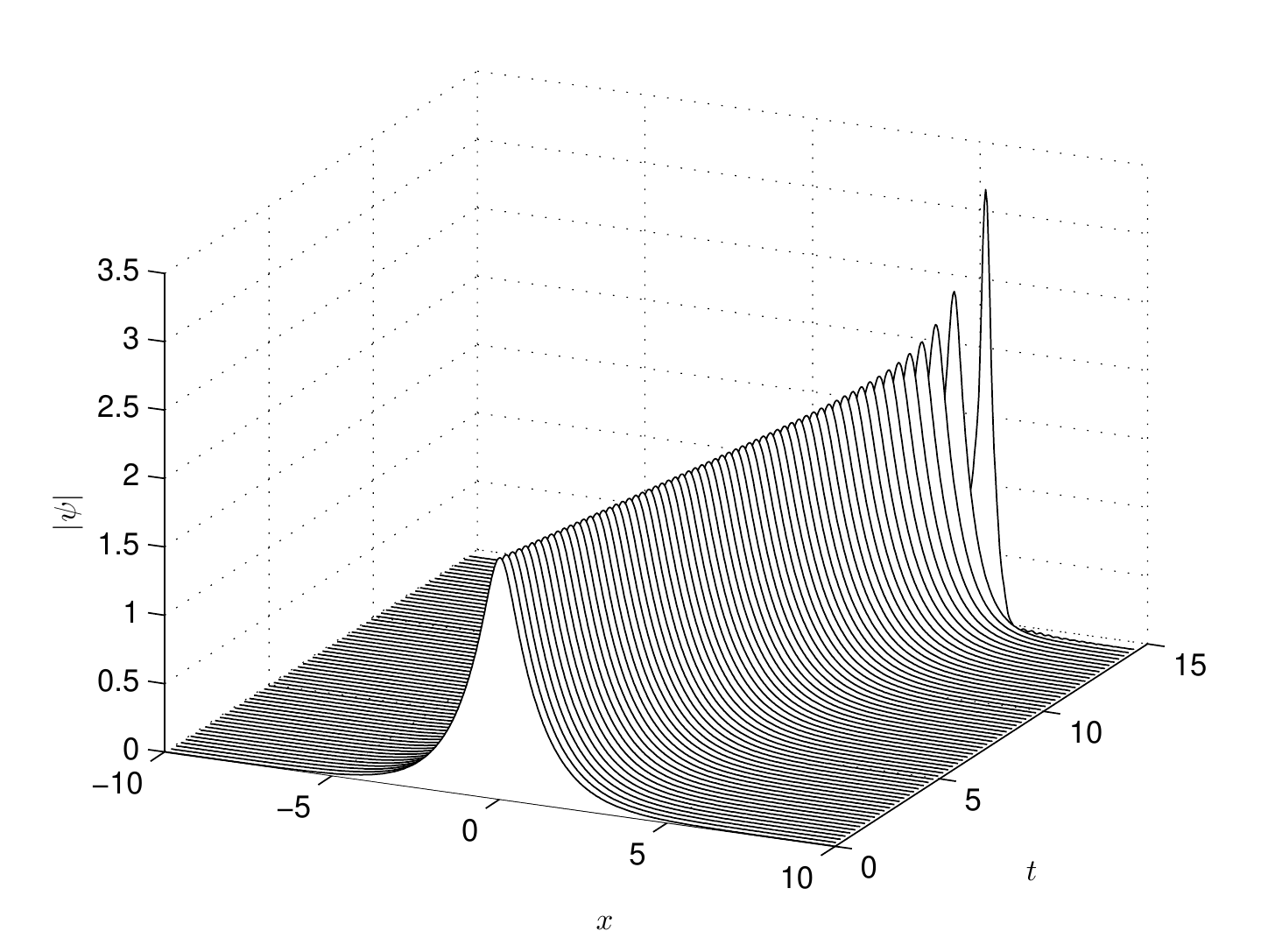}\\
  \caption{Evolution of a perturbed orbitally unstable soliton,
    $\omega =1$ and $c=.2$, with initial condition
    \eqref{e:sim_ic}.}\label{fig:orb_instable}
\end{figure}

\appendix

\section{Auxiliary Calculations}
In this section, we present certain integral relations that are
helpful in studying the determinant and trace of $d''(\omega,c)$.  In
the following, we denote $\kappa =\sqrt{4\omega-c^2} >0$ and rewrite
the solitary solution \eqref{eqn:def_phi} as $ \varphi(x)^{2\sigma} =
f(\omega,c) h(\omega,c;x)^{-1},$with
\begin{equation*}
  f(\omega,c)=\frac{(\sigma+1)\kappa^2 }{2\sqrt{\omega}} \ ,\;\;\;
  h(x;\sigma;\omega,c)=\cosh(\sigma\kappa x)-\frac{c}{2\sqrt{\omega}}.
\end{equation*}
We also rewrite the functionals $Q$, $P$ defined in
\eqref{eqn:Mass_dnls}-\eqref{eqn:Momentum_dnls} and their derivatives
in terms of $h$ and $f$:
\begin{align}
  \label{eqn:Q_phi}
  Q &= \frac{1}{2}\int_0^\infty|\varphi|^2 dx=f^{\frac{1}{\sigma}}\int_0^\infty h^{-\frac{1}{\sigma}} dx,\\
  \label{eqn:P_phi}
  \begin{split}
    P &=-\frac{c}{2}\int_0^\infty \varphi^2 dx+\frac{1}{2\sigma+2}\int_0^\infty\varphi^{2\sigma+2}dx\\
    &=-\frac{c}{2}f^{\frac{1}{\sigma}}\int_0^\infty
    h^{-\frac{1}{\sigma}}dx+\frac{1}{2\sigma+2}f^{\frac{\sigma+1}{\sigma}}\int_0^\infty
    h^{-\frac{\sigma+1}{\sigma}} dx.
  \end{split}
\end{align}
\begin{align}
  \label{eqn:Q_c}
  \partial_c Q
  &=\frac{1}{\sigma}f^{\frac{1-\sigma}{\sigma}}f_c\int_0^\infty
  h^{-\frac{1}{\sigma}}
  dx-\frac{1}{\sigma}f^{\frac{1}{\sigma}}\int_0^\infty
  h^{-\frac{\sigma+1}{\sigma}}h_c dx,\\
  \label{eqn:Q_omega}
  \partial_\omega Q&
  =\frac{1}{\sigma}f^{\frac{1-\sigma}{\sigma}}f_\omega\int_0^\infty
  h^{-\frac{1}{\sigma}}
  dx-\frac{1}{\sigma}f^{\frac{1}{\sigma}}\int_0^\infty
  h^{-\frac{\sigma+1}{\sigma}}h_\omega dx,
\end{align}
\begin{align}
  \label{eqn:P_c}
  \begin{split}
    \partial_c P =&-\frac{1}{2}f^{\frac{1}{\sigma}}\int_0^\infty
    h^{-\frac{1}{\sigma}} dx -\frac{c}{2
      \sigma}f^{\frac{1-\sigma}{\sigma}}f_c\int_0^\infty
    h^{-\frac{1}{\sigma}} dx\\
    &\quad + \frac{c}{2\sigma}f^{\frac{1}{\sigma}}\int_0^\infty
    h^{-\frac{\sigma+1}{\sigma}}h_c dx+
    \frac{1}{2\sigma}f^{\frac{1}{\sigma}}f_c\int_0^\infty
    h^{-\frac{1+\sigma}{\sigma}}dx\\
    &\quad -\frac{1}{2\sigma}f^{\frac{1+\sigma}{\sigma}}\int_0^\infty
    h^{-\frac{1+2\sigma}{\sigma}}h_c dx
  \end{split}\\
  \label{eqn:P_omega}
  \begin{split}
    \partial_\omega P=&-\frac{c}{2 \sigma}f^{\frac{1-\sigma}{\sigma}}f_\omega\int_0^\infty h^{-\frac{1}{\sigma}} dx+\frac{c}{2\sigma}f^{\frac{1}{\sigma}}\int_0^\infty h^{-\frac{\sigma+1}{\sigma}}h_\omega dx\\
    &\quad+ \frac{1}{2\sigma}f^{\frac{1}{\sigma}}f_\omega\int_0^\infty
    h^{-\frac{1+\sigma}{\sigma}}dx
    -\frac{1}{2\sigma}f^{\frac{1+\sigma}{\sigma}}\int_0^\infty
    h^{-\frac{1+2\sigma}{\sigma}}h_\omega dx,
  \end{split}
\end{align}
where
\begin{align*}
  f_c& = -\frac{c(1+\sigma)}{\sqrt{\omega}}, \;\;\;
  f_\omega=\frac{(1+\sigma)(4\omega+c^2)}{4\omega^{3/2}},\cr
  h_c&=-\frac{\sigma c}{\kappa} x\sinh(\sigma\kappa x)
  -\frac{1}{2}\omega^{-1/2},\; h_\omega=\frac{2\sigma}{\kappa}
  x\sinh(\sigma \kappa x) +\frac{c}{4}\omega^{-\frac{3}{2}}.
\end{align*}
The expressions in \eqref{eqn:Q_c}-\eqref{eqn:P_omega} involve various
integrals. The next lemmas show that all of them can be expressed
simply in terms of
\[
\alpha_n =\int_0^\infty h^{-\frac{1}{\sigma}-n} dx.
\]
First, we have
\begin{lemma}
  \label{le:integration_alpha_n}
  \begin{equation}
    \label{eqn:int_rel_5}
    \alpha_2 =
    \frac{4\omega}{(\sigma+1)\kappa^2} \alpha_0 +
    \frac{2c\sqrt{\omega}(2+\sigma)}{(\sigma+1) \kappa^2} \alpha_1
  \end{equation}
\end{lemma}

\begin{proof}
  We first rewrite $\alpha_0$, and then integrate by parts:
  \begin{equation*}
    \begin{split}
      \alpha_0 = \int_0^{\infty} h^{-\frac{1}{\sigma}-1} h
      dx&=\frac{1}{\sigma\kappa}\int_0^{\infty}
      h^{-\frac{1}{\sigma}-1}(\sinh(\sigma\kappa x))'dx- \alpha_1 \frac{c}{2\sqrt{\omega}}\\
      &=\frac{(\sigma+1)}{\sigma^2}\int_0^{\infty}
      h^{-\frac{1}{\sigma}-2}(\sinh^2(\sigma\kappa x))dx- \alpha_1 \frac{c}{2\sqrt{\omega}}\\
      &=\frac{(\sigma+1)}{\sigma^2}\int_0^{\infty}
      h^{-\frac{1}{\sigma}-2}((h+\frac{c}{2\sqrt{\omega}})^2-1)dx-
      \alpha_1 \frac{c}{2\sqrt{\omega}}.
    \end{split}
  \end{equation*}
  Regrouping the terms in this last expression, we obtain
  \eqref{eqn:int_rel_5}.
\end{proof}

\begin{lemma}
  \label{le:integration_relation}
  We have the following relations:
  \begin{eqnarray}
    \label{eqn:int_rel_1}
    &&   \int_0^\infty h^{-\frac{1}{\sigma}-2} h_c dx=
    -\frac{1}{2\sqrt{\omega}} \alpha_2
    - \frac{c\sigma}{(\sigma+1)\kappa^2} \alpha_1,\\
    && \label{eqn:int_rel_2}
    \int_0^\infty h^{-\frac{1}{\sigma}-1} h_c dx = -\frac{1}{2\sqrt{\omega}} \alpha_1 - \frac{c\sigma}{\kappa^2} \alpha_0,\\
    && \label{eqn:int_rel_3}
    \int_0^\infty h^{-\frac{1}{\sigma}-2} h_\omega dx = \frac{c}{4\omega^{3/2}} \alpha_2 + \frac{2\sigma}{(\sigma+1)\kappa^2} \alpha_1,\\
    &&    \label{eqn:int_rel_4}
    \int_0^\infty h^{-\frac{1}{\sigma}-1} h_\omega dx = \frac{c}{4\omega^{3/2}} \alpha_1 + \frac{2\sigma}{\kappa^2} \alpha_0.
  \end{eqnarray}
\end{lemma}

\begin{proof}
  By integration by parts, and $n$ integer,
  \begin{equation*}
    \int_0^\infty h^{-\frac{1}{\sigma}-n}h_c dx
    = \frac{ c}{\kappa^2 (-\frac{1}{\sigma}-n+1)}  \int_0^\infty
    h^{-\frac{1}{\sigma}-n+1} dx
    -\frac{1}{2\sqrt{\omega}}\int_0^\infty h^{-\frac{1}{\sigma}-n}
    dx,
  \end{equation*}
  Choosing $n=2,1$, we get \eqref{eqn:int_rel_1} and
  \eqref{eqn:int_rel_2}.
  The relations \eqref{eqn:int_rel_3} and \eqref{eqn:int_rel_4} are
  obtained from \eqref{eqn:int_rel_1}, \eqref{eqn:int_rel_2} and
  $\displaystyle{h_\omega =-\frac{2}{c}h_c
    -\frac{\kappa^2}{4\omega^{3/2}c}.}$
\end{proof}

Using Lemmas \ref{le:integration_alpha_n} and
\ref{le:integration_relation}, we have:
\begin{lemma}
  \label{le:QcQw_PcPw}
  Denoting $\tilde \kappa = 2^{-\frac{1}{\sigma}-2} \sigma^{-1}
  (1+\sigma)^{\frac{1}{\sigma}}\kappa^{2(\frac{1}{\sigma}-1)}
  \omega^{-\frac{1}{2\sigma}-\frac{1}{2}} $, we have
  \begin{align*}
    \partial_c Q &= 2 \tilde \kappa
    \left[2 c(\sigma-2)\omega^{1/2}\alpha_0+\kappa^2\alpha_1 \right]\\
    \partial_\omega Q &= \tilde \kappa \omega^{-1} \left[ \left(2 c^2
        -8 (\sigma-1)
        \omega\right)\omega^{1/2}\alpha_0 -\kappa^2 c \alpha_1 \right]\\
    \partial_c P &=\tilde \kappa \left[\left(2 c^2 -8 (\sigma-1)
        \omega\right)\omega^{1/2} \alpha_0 -
      \kappa^2 c \alpha_1\right]\\
    \partial_\omega P &= 2\tilde \kappa \left[2c(\sigma-2)
      \omega^{1/2} \alpha_0 +\kappa^2\alpha_1\right].
  \end{align*}

\end{lemma}
This is used in Section \ref{sec:concavity_d} to obtain
\eqref{e:detHessd}.

\end{document}